\documentclass[journal,onecolumn]{IEEEtran}

\ifCLASSINFOpdf

\else

\fi

\usepackage{graphicx}	
\usepackage{amsmath}
\usepackage{mathtools,amssymb}
\usepackage{amsthm}
\usepackage{bbm}
\usepackage{mathtools}
\usepackage{tabularx,algorithm}	
\usepackage{color}
\usepackage{algpseudocode}
\usepackage{bm}
\usepackage{tikz}
\usepackage{epstopdf}
\usepackage{cite}
\usepackage{amsfonts}
\usepackage{dsfont}
\usepackage{amssymb}
\usepackage{algorithm}
\usepackage{epsfig}
\usepackage[labelformat=simple]{subcaption}

\newcommand{\overbar}[1]{\mkern 2.5mu\overline{\mkern-2.5mu#1\mkern-2.5mu}\mkern 2.5mu}

\newtheorem{theorem}{Theorem}
\newtheorem{definition}{Definition}

\newtheorem{lemma}{Lemma}

\newtheorem{example}{Example}
\newtheorem{assumption}{Assumption}
\newtheorem{corollary}{Corollary}

\usetikzlibrary{positioning}
\usetikzlibrary{decorations.pathreplacing}
\usetikzlibrary{shapes,arrows}

\begin{document}

\title{Consistent and Asymptotically Efficient Localization from Range-Difference Measurements}
\author{Guangyang Zeng, Biqiang Mu, Ling Shi, Jiming Chen, and Junfeng Wu
\thanks{G. Zeng is with the School of Data Science, Chinese University of Hong Kong, Shenzhen, Shenzhen, P. R. China, and also with the Department of Automation, University of Science and Technology of China, Hefei, P. R. China.
		{\tt\small zengguangyang@cuhk.edu.cn}.}
	\thanks{B. Mu is with Key Laboratory of Systems and Control, Institute of Systems Science, Academy of Mathematics and Systems Science, Chinese Academy of Sciences, Beijing 100190, China.
		{\tt\small bqmu@amss.ac.cn}.}
	\thanks{L. Shi is with the Department of Electronic
and Computer Engineering, Hong Kong University of Science and
Technology, Hong Kong.
		{\tt\small eesling@ust.hk}.}
\thanks{J. Chen is with the College of Control Science and Engineering and the State Key Laboratory of Industrial Control Technology, Zhejiang University, Hangzhou 310027, P. R. China.
		{\tt\small cjm@zju.edu.cn}.}
  \thanks{J. Wu is with the School of Data Science, Chinese University of Hong Kong, Shenzhen, Shenzhen, P. R. China.
		{\tt\small junfengwu@cuhk.edu.cn}.}
}

\maketitle

\begin{abstract} 
	We consider signal source localization from range-difference measurements. First, we give some readily-checked conditions on measurement noises and sensor deployment to guarantee the asymptotic identifiability of the model and show the consistency and asymptotic normality of the maximum likelihood (ML) estimator.
	Then, we devise an estimator that owns the same asymptotic property as the ML one. Specifically, we prove that the negative log-likelihood function converges to a function, which has a unique minimum and positive definite Hessian at the true source's position. Hence, it is promising to execute local iterations, e.g., the Gauss-Newton (GN) algorithm, following a consistent estimate. The main issue involved is obtaining a preliminary consistent estimate. To this aim, we construct a linear least-squares problem via algebraic operation and constraint relaxation and obtain a closed-form solution. We then focus on deriving and eliminating the bias of the linear least-squares estimator, which yields an asymptotically unbiased (thus consistent) estimate. Noting that the bias is a function of the noise variance, we further devise a consistent noise variance estimator that involves $3$-order polynomial rooting. Based on the preliminary consistent location estimate, a one-step GN iteration suffices to achieve the same asymptotic property as the ML estimator. Simulation results demonstrate the superiority of our proposed algorithm in the large sample case.
\end{abstract}

%\begin{IEEEkeywords}
%Range measurements, TOA localization, Two-step localization, Large-sample analysis
%\end{IEEEkeywords}

% For peer review papers, you can put extra information on the cover
% page as needed:
% \ifCLASSOPTIONpeerreview
% \begin{center} \bfseries EDICS Category: 3-BBND \end{center}
% \fi
%
% For peerreview papers, this IEEEtran command inserts a page break and
% creates the second title. It will be ignored for other modes.
\IEEEpeerreviewmaketitle

\section{Introduction} \label{section_introduction}
Signal source localization refers to calculating a source's spatial coordinates with respect to a specific coordinate system by using some sensors' measurements. It serves as a fundamental technology in extensive location-aware applications, ranging from navigation systems~\cite{wang2017network}, battlefield monitoring~\cite{meng2020multiple} to social networks~\cite{yang2020lbsn2vec++} and ads recommendation~\cite{huang2018location}. Range difference, usually calculated from the time difference of arrival (TDOA), is a widely used measurement for source localization, which can achieve high localization accuracy~\cite{bishop2010optimality,sun2018solution,zeng2022localizability,xiong2021cooperative,shen2010fundamental,han2015performance}. 
It does not require synchronization between the source and sensors and thus can be used in asynchronized or non-cooperative scenarios~\cite{zekavat2011handbook}.

TDOA is tightly related to the direction of arrival (DOA), and TDOA localization (considered in this paper) should be distinguished from DOA estimation. DOA estimation refers to inferring the direction of a source, which is a widely studied topic. It generally assumes a far-field model where the source is far from the sensor array and the incident waves are parallel to each other~\cite{zheng2023coarray}. Then, DOA can be inferred from phase difference (i.e., TDOA) measurements. For example, Zhou \emph{et al.}~\cite{zhou2023structured} proposed a novel sparse array DOA estimation algorithm via structured correlation reconstruction, which can guarantee general applicability and a more flexible constraint on the array configuration. Zheng \emph{et al.}~\cite{zheng2023coarray2} proposed a coarray tensor DOA estimation algorithm for multi-dimensional structured sparse arrays and investigated an optimal coarray tensor structure for source identifiability enhancement. However, when TDOA measurements are utilized in source localization where both the direction and distance of a source need to be estimated, a near-field model is required. The location of the source is obtained by finding the intersection of several hyperbolas or hyperboloids defined by TDOA measurements. If the source is far from the sensor array, the estimation of the distance is not reliable~\cite{sun2018solution}. 

Maximum likelihood (ML) and least squares (LS) are the most common criteria for problem formulation in parameter inference. When measurement noises are i.i.d. Gaussian random variables, the LS criterion is equivalent to the ML one. Due to the non-linear property of range-difference measurements, the resulting ML and LS problems are non-convex, whose global minimizer is difficult to obtain. When utilizing iterative local search methods, an appropriate initial value is needed, otherwise, it may converge to local minima~\cite{torrieri1984statistical,mensing2006positioning}. Most works transformed the original ML and LS problems into some solvable problems via various methods, e.g., linear approximation~\cite{shi2020acoustic,chang2018surveillance}, semidefinite programming (SDP)~\cite{lui2008semidefinite,xu2010reduced,zou2018semidefinite,wang2019convex,wang2011importance}, and spherical model~\cite{beck2008exact,zeng2022localizability,amiri2018efficient,qu2017iterative}. It is noteworthy that since the transformed problems are generally not equivalent to the original, although they can be optimally solved, their global minimizer does not necessarily coincide with the original one. To the best of our knowledge, there is no algorithm that can theoretically guarantee obtaining the global minimizer of the ML and LS problems for TDOA localization.  

The ML estimator is optimal in the statistical sense that under some regularity conditions, it is consistent and asymptotically normal\footnote{Unless otherwise specified, ``asymptotically'' and ``asymptotic'' mean that the number of measurements goes to infinity.}. We note that testing these general regularity conditions in range-difference-based localization settings is nontrivial. In this paper, we provide some readily-checked conditions to guarantee the consistency and asymptotic normality of the ML estimator. 
Moreover, we claim that although the ML problem cannot be directly solved, we can devise an estimator that owns the same asymptotic property as the ML one. 
{We prove that the negative log-likelihood function uniformly converges to a function with the true source's location being its unique minimizer. In addition, we show that the function is convex in a neighborhood (called the attraction region) around the minimizer. Therefore, before using local iterative methods to obtain a precise solution, the key is to obtain a consistent estimate that will fall into the attraction region as measurements increase. }
To this aim, we construct a linear least-squares estimator and analyze its asymptotic bias for bias elimination. 
There are some works that derive biases and give bias-reduced solutions. Ho~\cite{ho2012bias} proposed two methods to reduce the bias of the well-known algebraic explicit solution known as CFS~\cite{chan1994simple}. Wang et al.~\cite{wang2015bias} derived the bias for a non-linear weighted least-squares (WLS) estimator using the Karush–Kuhn–Tucker (KKT) optimality conditions and first-order Taylor-series expansion. Chen and Ho~\cite{chen2013achieving} analyzed the asymptotic bias for the squared range-difference formulation where the ``asymptotic'' here means that the noise intensity is sufficiently small. Zhang et al.~\cite{zhang2021efficient} approximated the bias from the ML estimation and developed a bias-reduced iterative constrained WLS algorithm. 
We remark that the existing bias-reduced methods mostly analyzed the bias at most up to the second-order statistics of the measurement noises under the assumption that the noise intensity is small. As a result, the bias-reduced solutions are usually inaccurate in the large noise region---they may even be worse than the estimates before bias reduction. 

In this paper, we propose an exact asymptotic bias elimination method for the large sample case with no restriction on noise intensity. Specifically, we square the original measurement model and formulate an ordinary LS problem via constraint relaxation. However, since the transformed noise term does not have zero mean, and the regressor and regressand are correlated, the closed-form solution is biased. Then, with the prior knowledge of noise variance or a consistent estimate of it, we make the modified noise term have zero mean and eliminate the asymptotic correlation between the regressor and regressand. In virtue of the above bias elimination techniques, the resulting solution is asymptotically unbiased and thus is consistent. Further, taking the preliminary consistent estimate as the initial value, Gauss-Newton (GN) iterations are used to search for the ML solution. 
The resulting estimator is asymptotically efficient, i.e., its mean square error (MSE) asymptotically reaches the theoretical lower bound --- Cramer-Rao lower bound (CRLB). This appealing property makes it achieve highly accurate estimation when the measurement number is large. Note that some of the existing
localization systems have a high speed of measurements, e.g., the sampling rate of practical Ultra-wideband (UWB) systems can reach 2.3 kHz~\cite{grossiwindhager2019snaploc}. Therefore, the proposed asymptotically efficient estimators can play a valuable role in these systems, especially when the source is static, and a large sample of measurements can be utilized.

We note that the asymptotically efficient localization based on range measurements has been investigated in~\cite{zeng2022globally}. The gap between range-difference measurement and range measurement is nontrivial. The range model is a single norm function, while the range-difference model has the form of the subtraction of two norm functions. As a result, model transformation in TDOA localization leads to the correlation between the regressor and regressand, which makes bias elimination more complicated. While no such problem arises in range-based localization. 
We also note that Wang \emph{et al.}~\cite{wang2017solving} investigated the phase retrieval problem where an unknown vector is recovered from a system of quadratic equations. The authors proposed a novel algorithm, termed truncated amplitude flow, that adapts the amplitude-based empirical loss function and proceeds in two stages. It was proved that as the number of quadratic equations increases, the estimate converges to the true vector (up to a sign) with high probability, and the computational complexity grows linearly. This is similar to the consistent property of our algorithm. The main difference between the work and ours lies in the objective function. The objective in~\cite{wang2017solving} contains a modulus of the inner product of a random vector and the unknown vector, while that in our problem includes the subtraction of two norm functions. This ultimately leads to the fact that the sign of the unknown vector cannot be recovered in the phase retrieval problem while being identifiable in the TDOA localization problem. In addition, it also gives rise to nontrivial differences between these two problems in terms of algorithm development and analysis. 

In summary, the main contributions of this paper are listed as follows:

\begin{enumerate}
	\item  [$(i).$] We give some conditions on measurement noises and sensor deployment to ensure the asymptotic identifiability of the model and prove the consistency and asymptotic normality of the ML estimator. These conditions are specifications of regularity conditions in the context of TDOA-based localization, which can be readily checked. 
	\item  [$(ii).$] We associate noise variance estimation with a maximum-eigenvalue-related equation rooting problem and prove that the smallest root is a consistent estimate of the measurement noise variance. Moreover, the problem is converted into a 3-order polynomial rooting problem which can be efficiently solved. 
	\item  [$(iii).$] We propose a closed-form consistent localization method via precise (in the asymptotic sense) bias elimination, where a noise variance estimate is utilized. Based on the preliminary consistent solution, a one-step GN iteration is sufficient to achieve the same asymptotic property as the ML estimator.  
\end{enumerate}

{\bf Notations:} We use bold lowercase letters to denote vectors, e.g., ${\bf x}$, $\bf y$, $\bf z$, and bold uppercase letters for matrices, e.g., ${\bf X}$, $\bf Y$, $\bf Z$. For a vector ${\bf x}$, $[{\bf x}]_i$ presents its $i$-th element, and $\|{\bf x}\|$ denotes its $2$-norm. For a matrix ${\bf X}$, $[{\bf X}]_{ij}$ presents its element that locates at the $i$-th row and $j$-th column. The identity matrix of size $n$ is represented as ${\bf I}_n$. The all-zeros matrix of size $n \times m$ is denoted as ${\bf 0}_{n \times m}$. For two vectors ${\bf x},{\bf y} \in \mathbb R^n$, ${\bf x} \preceq {\bf y}$ denotes the pointwise inequality. Let $p=(p_i)_{i\in\mathbb N}$ and $q=(q_i)_{i\in\mathbb N}$ be two sequences of real numbers. When ${i\in\mathbb N}$ is clear from the context, we will omit the subscript and write $(p_i)$ as a shorthand of  $(p_i)_{i\in\mathbb N}$. If $t^{-1} \sum_{i=1}^{t} p_i q_i$ converges to a real number its limit $\left\langle p,q \right\rangle_t $ will be called the tail product of $p$ and $q$.
We call $\|p\|_t=\sqrt{\langle p,p
	\rangle_t}$,
if it exists, the tail norm of $p$. For a sequence $p=(p_i)$ and a scalar $c$, $p-c$ produces a sequence, of which the $i$-th element is $p_i-c$.
For a cumulative distribution function $F_{\mu}$, $\mu$ is the measure induced from $F_{\mu}$, and $\mathbb E_{\mu}[\cdot]$ denotes taking the expectation with respect to $\mu$. The notation $X_m=O_p(c_m)$ means that the sequence of $X_m/c_m$ is stochastically bounded, and $X_m=o_p(c_m)$ means that the sequence of $X_m/c_m$ converges to zeros in probability. 

\section{Consistency and Asymptotic Normality of the ML Estimator} \label{ML_estimator}
\subsection{Problem formulation}
Fig.~\ref{fig:figure_TDoA} shows the illustration of source localization using a sensor array, where ${\bf x}^o \in \mathbb R^n$ denotes the coordinates of the source and ${\bf a}_i \in \mathbb R^n, i=0,1,\ldots,m$ are the coordinates of the sensors. In particular, ${\bf a}_0$ is the reference sensor's coordinates, and without loss of generality, we set it to $0$. The range-difference measurement between sensor $i$ and the reference one has the following equation:
\begin{equation}\label{range_difference_measurements}
	d_i=\|{\bf a}_i-{\bf x}^o\|-\|{\bf x}^o\|+r_i, \;\; i=1,2,\ldots,m,
\end{equation}
where $r_i$ represents the measurement noise. For the measurement noises, we make the following assumption.
\begin{assumption} \label{assumption_1}
	The measurement noises $r_i,i=1,\ldots,m$ are $i.i.d.$ Gaussian noises with zero mean and finite variance $\sigma^2$.
\end{assumption}
In wireless ranging techniques including UWB and LoRa, the measurement variance largely depends on the bandwidth of the channel and can be viewed as a constant (irrespective of the distance) when full bandwidth is used~\cite{jiang2023efficient,bellusci2008model}. Therefore, Assumption~\ref{assumption_1} which assumes that each sensor has the same measurement variance is realistic. 

\begin{figure} [htbp]
	\centering
	\includegraphics[width=0.3\textwidth]{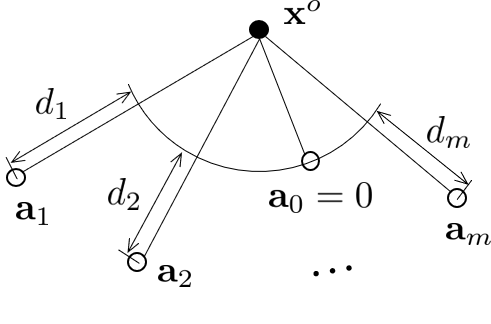}
	\caption{Illustration of range-difference measurements in the $2$D case. The solid dot ``$\bullet$'' represents the radiating source and the hollow dots ``$\circ$'' represent the sensors.}
	\label{fig:figure_TDoA}
\end{figure}

Based on~\eqref{range_difference_measurements}, the ML problem is given as follows:
\begin{equation}\label{ML_problem}
	\hbox{(\textbf{ML}):}~~~~~\mathop{\rm minimize}_{{\bf x} \in\mathbb R^n} \frac{1}{m}\sum_{i=1}^{m} (d_i+\|{\bf x}\|-\|{\bf a}_i-{\bf x}\|)^2.
\end{equation}
Denote an optimal solution to problem~\eqref{ML_problem} as an ML estimate. We know that ML estimators are consistent and asymptotically normal under some regularization conditions~\cite{casella2021statistical}. However, it is worth noting that the regularization conditions are highly abstract and not easy to check in range-difference-based localization settings. In the rest of this section, we will present some readily-checked conditions associated with the source's coordinates and sensor deployment to guarantee the consistency and asymptotic normality of ML estimates. 

\subsection{Asymptotic properties of the ML estimate} 
Assumption~\ref{assumption_1} is about the measurement noises. In the following, we give assumptions on the source's coordinates and sensor deployment. 
\begin{assumption} \label{compact_and_bounded}
	The source's coordinates ${\bf x}^o$ belongs to a compact set $\mathcal X$, the coordinates of sensors ${\bf a}_i,i=0,\ldots,m$ belong to a bounded set $\mathcal A$, and for each $i$, ${\bf a}_i \neq {\bf x}^o$.
\end{assumption}

\begin{definition}
	The sample distribution function $F_m$ of a sequence $({\bf z}_1,{\bf z}_2,\ldots)$ in $\mathbb R^n$ is defined as $F_m({\bf z})=\#/m$ where $\#$ is the number of vectors in the subsequence $({\bf z}_1,\ldots,{\bf z}_m)$ that satisfy ${\bf z}_i \preceq {\bf z}$.
\end{definition}
\begin{assumption} \label{convergence_of_sample_distribution}
	The sample distribution function $F_m$ of the sequence $({\bf a}_1,{\bf a}_2,\ldots)$ converges to a distribution function $F_{\mu}$, i.e., $\lim\limits_{m \rightarrow \infty} F_m({\bf a})=F_{\mu}({\bf a})$, for all ${\bf a} \in \mathbb R^n$. 
\end{assumption}

We denote the probability measure generated by $F_{\mu}$ as $\mu$. In what follows, we give two examples of sensor deployments that satisfy Assumption~\ref{convergence_of_sample_distribution}. 
\begin{example} \label{example_random_vector}
	When ${\bf a}_i, i=1,\ldots,m$ are independent realizations of some random vectors with identical distribution function $F_{\mu}$, we have $\lim\limits_{m \rightarrow \infty} F_m({\bf a})=F_{\mu}({\bf a})$ for all ${\bf a} \in \mathbb R^n$. 
\end{example}

\begin{example} \label{example_fix_sensors}
	Suppose the number of sensors $M$ (excluding the reference one) is fixed, and each sensor makes $T$ i.i.d. measurements. In this manner, a total of $MT$ TDOA measurements can be used. This setting is realistic when the object is static or the sampling of the TDOA measurements is sufficiently fast compared to the object's motion. In this setup, as $T$ goes to infinity, $F_m$ converges to $F_{\mu}$, where $\mu({\bf a}_i)=1/M$ for each $i$.
\end{example}

For the asymptotic distribution $F_{\mu}$ to which the sample distribution of sensors' coordinates converge, we make the following assumption. 

\begin{assumption} \label{deployment_of_sensor2}
	There does not exist any subset $\mathcal S \subset \mathbb R^n$ listed below such that $\mu(\mathcal S)=1$:
	\begin{enumerate}
		\item  [$(i).$] A hyperbola (resp. hyperboloid), which passes through ${\bf a}_0$ and has a focus being ${\bf x}^o$, for $n=2$ (resp. $n=3$).
		\item  [$(ii).$] A line (resp. plane) for $n=2$ (resp. $n=3$).
	\end{enumerate}
\end{assumption}
Sometimes Assumption~\ref{deployment_of_sensor2} is hard to check since the true source's position ${\bf x}^o$ is unknown. In what follows, we give a stricter assumption that implies Assumption~\ref{deployment_of_sensor2}, which requires the notions of conic sections and quadric surfaces.

\begin{assumption} \label{deployment_of_sensor}
	There does not exist a conic section (resp. quadric surface) $\mathcal S$ for $n=2$ (resp. $n=3$) such that $\mu(\mathcal S)=1$.
\end{assumption}
A conic section is a curve obtained as the intersection of the surface of a cone with a plane, and any point $\bf a$ on it satisfies a quadratic equation in the form
\begin{equation*}
	c_1[{\bf a}]_1^2+c_2[{\bf a}]_1[{\bf a}]_2+c_3[{\bf a}]_2^2+c_4[{\bf a}]_1+c_5[{\bf a}]_2+c_6=0,
\end{equation*}
for some real coefficients $c_1,\ldots, c_6$.
A quadric surface is a generalization of the conic section and any point $\bf a$ on it satisfies a quadratic equation in the form
\begin{equation*}
	c_1[{\bf a}]_1^2  +c_2[{\bf a}]_2^2+c_3[{\bf a}]_3^2+c_4[{\bf a}]_1[{\bf a}]_2+c_5[{\bf a}]_1[{\bf a}]_3  +c_6[{\bf a}]_2[{\bf a}]_3+c_7[{\bf a}]_1+c_8[{\bf a}]_2+c_9[{\bf a}]_3+c_{10}=0,
\end{equation*}
for some real coefficients $c_1,\ldots,c_{10}$.
According to~\cite{coxeter1967geometry},
five points in general position uniquely determine a conic section, and nine points in general position uniquely determine a quadric surface. Therefore, Assumption~\ref{deployment_of_sensor} implies that the measure $\mu$ cannot concentrate on five points (resp. nine points) for $n=2$ (resp. $n=3$). To test whether a set of $k$ sensors, indexed as ${\bf a}_1,\ldots, {\bf a}_k$, form a conic section, we can test the rank of the matrix 
$[{\bf v}_1,\ldots,{\bf v}_k]^\top$, where  ${\bf v}_i^\top=\left[[{\bf a}_i]_1^2,[{\bf a}_i]_1[{\bf a}_i]_2,[{\bf a}_i]_2^2,[{\bf a}_i]_1,[{\bf a}_i]_2,1 \right]$. If the rank is six, then the sensors do not locate on a conic section. 
Note that Assumption~\ref{deployment_of_sensor} is an asymptotic condition. To verify it for $n=2$, we can test whether the rank of the matrix $\mathbb E_{\mu} [{\bf v}({\bf a}) {\bf v}({\bf a})^\top]$ is six, where ${\bf v}({\bf a})^\top=\left[[{\bf a}]_1^2,[{\bf a}]_1[{\bf a}]_2,[{\bf a}]_2^2,[{\bf a}]_1,[{\bf a}]_2,1 \right]$ and $\mathbb E_{\mu}$ is taken over $\bf a$ with respect to $\mu$. If the rank is six, then Assumption~\ref{deployment_of_sensor} holds. The verification for $n=3$ is similar. 
Since a hyperbola (resp. hyperboloid) and a line (resp. plane) are both (degenerate) conic sections (resp. quadric surfaces), Assumption~\ref{deployment_of_sensor} is a sufficient condition for Assumption~\ref{deployment_of_sensor2}. Hence, we can verify Assumption~\ref{deployment_of_sensor}, instead of directly checking Assumption~\ref{deployment_of_sensor2}. How to verify all of the proposed assumptions will be illustrated in the simulation part.

% \begin{assumption} \label{deployment_of_sensor}
% 	There does not exist a hyperbola $\mathcal S$ for 2D localization (a hyperboloid $\mathcal S$ for 3D localization) such that $\mu(\mathcal S)=1$.
% \end{assumption}
% \begin{assumption} \label{deployment_of_sensor2}
% 	The matrix $\Phi^\top \Phi /m$, where $\Phi=\begin{bmatrix}
% 	    [a_1]_1^2 & [a_1]_1 [a_1]_2 & [a_1]_2^2 & [a_1]_1 & [a_1]_2 & 1 \\
%      \vdots & \vdots & \vdots & \vdots & \vdots & \vdots \\
%      [a_m]_1^2 & [a_m]_1 [a_m]_2 & [a_m]_1^2 & [a_m]_1 & [a_m]_2 & 1 
% 	\end{bmatrix}$, is non-singular in 2D case.
% \end{assumption}

Let $f_i({\bf x}):=\|{\bf a}_i-{\bf x}\|-\|{\bf x}\|$ and $f({\bf x}):=\left(f_i({\bf x}) \right) $. Assumption~\ref{convergence_of_sample_distribution} along with Assumption~\ref{compact_and_bounded} ensure that the tail norm $\left\| f({\bf x})-f({\bf x}^o) \right\|_t$ exists for all ${\bf x} \in \mathbb R^{n}$. This is based on the Helly-Bray theorem~\cite{billingsley2013convergence} which presents the convergence of the sample mean of any bounded, continuous, and real-valued function. We remark that the tail products and tail norms involved in the rest of this paper all exist given Assumptions~\ref{convergence_of_sample_distribution} and~\ref{compact_and_bounded}. The following lemma is on the asymptotic identifiability of model~\eqref{range_difference_measurements}. 
\begin{lemma} \label{Lemma_asymptotically_localizable}
	Given Assumptions~\ref{assumption_1}-\ref{convergence_of_sample_distribution},~\ref{deployment_of_sensor2}(i), the source is asymptotically uniquely localizable. Or equivalently, $\left\| f({\bf x})-f({\bf x}^o) \right\|^2_t$ has a unique minimum at ${\bf x}={\bf x}^o$. 
\end{lemma}
\begin{proof}
	By definition, we have 
	\begin{equation*}
\left\| f({\bf x})-f({\bf x}^o) \right\|^2_t=\mathbb E_{\mu}\left[\left(\|{\bf a}-{\bf x}\|-\|{\bf x}\|-\|{\bf a}-{\bf x}^o\|+\|{\bf x}^o\| \right) ^2 \right], 
\end{equation*}	
	where $\mathbb E_{\mu}$ is taken over $\bf a$ with respect to $\mu$, and $\left\| f({\bf x}^o)-f({\bf x}^o) \right\|^2_t=0$. For any $\bf x$, define $\mathcal A_{\bf x}=\{{\bf a} \in \mathcal A \mid \|{\bf a}-{\bf x}\|-\|{\bf a}-{\bf x}^o\|=\|{\bf x}\|-\|{\bf x}^o\| \}$. Suppose there is an ${\bf x}' \neq {\bf x}^o$ such that $\mathbb E_{\mu}\left[\left(\|{\bf a}-{\bf x}'\|-\|{\bf x}'\|-\|{\bf a}-{\bf x}^o\|+\|{\bf x}^o\| \right) ^2 \right]=0$. Then, $\mu(\mathcal A_{{\bf x}'})=1$. Note that $\mathcal A_{{\bf x}'}$ is the hyperbola (hyperboloid) with ${\bf x}^o$ and ${\bf x}'$ being its foci. This contradicts Assumption~\ref{deployment_of_sensor2}$(i)$. Hence, $\left\| f({\bf x})-f({\bf x}^o) \right\|^2_t$ has a unique minimum at ${\bf x}={\bf x}^o$.
\end{proof}

In Theorem~\ref{convergence_of_lik} of Section~\ref{Two_step_estimator}, we show that the objective function in the ML problem~\eqref{ML_problem} converges to $\left\| f({\bf x})-f({\bf x}^o) \right\|^2_t+\sigma^2$. Since the ML estimate minimizes the objective function, Lemma~\ref{Lemma_asymptotically_localizable} guarantees that the ML estimate is consistent, i.e., as the number of measurements increases, it converges to the true value ${\bf x}^o$. Further, let $f_j'({\bf x})=( f_{ji}'({\bf x})) $ be a sequence with respect to $i$, where 
\begin{equation*}
	f_{ji}'({\bf x})=\frac{\partial f_i({\bf x})}{\partial [{\bf x}]_j},
\end{equation*}
and ${\bf M}({\bf x}) \in \mathbb R^{n \times n}$ be a matrix of which $[{\bf M}({\bf x})]_{jk}=\left\langle f_j'({\bf x}),f_k'({\bf x})\right\rangle_t $. Then the following lemma holds. 
\begin{lemma} \label{nonsingular_fisher}
	Given Assumptions~\ref{assumption_1}-\ref{convergence_of_sample_distribution},~\ref{deployment_of_sensor2}(ii), the matrix ${\bf M}({\bf x}^o)$ is non-singular. 
\end{lemma}
\begin{proof}
	Here, we give the proof of the 2D case. The argument of the 3D case is similar and will be omitted. For any ${\bf x} \neq 0$, define $\mathcal A^{\perp}_{\bf x}=\{ {\bf a} \in \mathcal A \mid {\bf x}^\top \left(\frac{{\bf x}^o-{\bf a}}{\|{\bf x}^o-{\bf a}\|}-\frac{{\bf x}^o}{\|{\bf x}^o\|}\right) = 0 \}$, and $\overbar {{\mathcal A}^{\perp}_{\bf x}}=\mathcal A \setminus \mathcal A^{\perp}_{\bf x}$, where $\setminus$ denotes set difference. Note that ${\bf x}^\top \frac{{\bf x}^o-{\bf a}}{\|{\bf x}^o-{\bf a}\|}={\bf x}^\top \frac{{\bf x}^o}{\|{\bf x}^o\|}$ implies the projection of $\bf x$ onto the vector ${\bf x}^o-{\bf a}$ is a constant. Therefore, the set $\mathcal A^{\perp}_{\bf x}$ is a subset of a line. In virtue of Assumption~\ref{deployment_of_sensor2}$(ii)$, we have that for any ${\bf x} \neq 0$, $\mu\left(\overbar {{\mathcal A}^{\perp}_{\bf x}} \right) >0$. Then we can decompose ${\bf M}({\bf x}^o)$ as
\begin{align*}
{\bf M}({\bf x}^o) = &\lim\limits_{m \rightarrow \infty} \frac{1}{m} \sum_{i=1}^{m} \nabla f_i({\bf x}^o) \nabla f_i({\bf x}^o)^\top \\
=&\mathbb E_{\mu}\left[\left(\frac{{\bf x}^o-{\bf a}}{\|{\bf x}^o-{\bf a}\|}-\frac{{\bf x}^o}{\|{\bf x}^o\|}\right) \left(\frac{{\bf x}^o-{\bf a}}{\|{\bf x}^o-{\bf a}\|}-\frac{{\bf x}^o}{\|{\bf x}^o\|}\right)^\top \right] \\
= &\underbrace{\int_{\mathcal A^{\perp}_{\bf x}} \left(\frac{{\bf x}^o-{\bf a}}{\|{\bf x}^o-{\bf a}\|}-\frac{{\bf x}^o}{\|{\bf x}^o\|}\right) \left(\frac{{\bf x}^o-{\bf a}}{\|{\bf x}^o-{\bf a}\|}-\frac{{\bf x}^o}{\|{\bf x}^o\|}\right)^\top d \mu({\bf a})}_{:= {\bf M}^{\perp}_{\bf x}({\bf x}^o)} \\
+ & \underbrace{\int_{\overbar {{\mathcal A}^{\perp}_{\bf x}}} \left(\frac{{\bf x}^o-{\bf a}}{\|{\bf x}^o-{\bf a}\|}-\frac{{\bf x}^o}{\|{\bf x}^o\|}\right) \left(\frac{{\bf x}^o-{\bf a}}{\|{\bf x}^o-{\bf a}\|}-\frac{{\bf x}^o}{\|{\bf x}^o\|}\right)^\top d \mu({\bf a})}_{:=\overbar {{{\bf M}}^{\perp}_{\bf x}}({\bf x}^o)}
\end{align*}	
	Since $\mu\left(\overbar {{\mathcal A}^{\perp}_{\bf x}} \right) >0$, we have $\overbar {{{\bf M}}^{\perp}_{\bf x}}({\bf x}^o) \neq 0$. Therefore,
	\begin{align*}
		{\bf x}^\top {\bf M}({\bf x}^o) {\bf x} &= {\bf x}^\top \left({\bf M}^{\perp}_{\bf x}({\bf x}^o)+ \overbar {{{\bf M}}^{\perp}_{\bf x}}({\bf x}^o)\right)  {\bf x} \\
		& = {\bf x}^\top \overbar {{{\bf M}}^{\perp}_{\bf x}}({\bf x}^o) {\bf x} \\
		& >0,
	\end{align*}
	which implies ${\bf M}({\bf x}^o)$ is positive definite and completes the proof.
\end{proof}

Now, we are on the point to depict the asymptotic property of the ML estimator that optimally solves~\eqref{ML_problem}. Denote the ML estimate as $\hat {\bf x}^{\rm ML}_{m}$. Under Assumptions~\ref{assumption_1}-\ref{convergence_of_sample_distribution},~\ref{deployment_of_sensor2}$(i\text{-}ii)$, the ML estimate $\hat {\bf x}^{\rm ML}_{m}$ enjoys the following consistency and asymptotic normality. The proof is straightforward by checking conditions in~\cite[Theorem 3]{jennrich1969asymptotic}.

\begin{theorem}[Consistency and asymptotic normality] \label{theorem_ML}
	Under Assumptions~\ref{assumption_1}-\ref{convergence_of_sample_distribution},~\ref{deployment_of_sensor2}(i-ii), we have $\hat {\bf x}^{\rm ML}_{m} \rightarrow {\bf x}^o$ with probability one as $m \rightarrow \infty$. Moreover,
	\begin{equation}\label{asymptotic_normality_LS}
		\sqrt{m}(\hat {\bf x}^{\rm ML}_{m}-{\bf x}^o) \rightarrow \mathcal N(0,\sigma^2 {\bf M}^{-1}({\bf x}^o))~~\text{as}~m \rightarrow \infty.
	\end{equation}
\end{theorem}

The matrix ${\bf M}({\bf x}^o)$ is tightly related to the Fisher information matrix $\bf F$ of model~\eqref{range_difference_measurements}. 
To derive the Fisher information matrix, recall that $d_i=\|{\bf a}_i-{\bf x}^o\|-\|{\bf x}^o\|+r_i$. Let ${\bf d}=[d_1,\ldots,d_m]^\top$. Since $r_i \sim \mathcal N(0,\sigma^2)$ are i.i.d., we obtain the log-likelihood function as follows
\begin{align*}
	\ell ({\bf d};{\bf x}^o) & = \log \left(\prod_{i=1}^{m} \frac{1}{\sqrt{2 \pi}\sigma}e^{-\frac{\left(d_i-\|{\bf a}_i-{\bf x}^o\|+\|{\bf x}^o\| \right)^2}{2\sigma^2}}\right) \\
	& =m \log \frac{1}{\sqrt{2 \pi}\sigma}-\sum\limits_{i=1}^{m} \frac{\left(d_i-\|{\bf a}_i-{\bf x}^o\|+\|{\bf x}^o\| \right)^2}{2\sigma^2},
\end{align*}
which gives
\begin{equation*}
	\frac{\partial \ell ({\bf d};{\bf x}^o)}{\partial {\bf x}^o}=-\frac{1}{\sigma^2}\sum\limits_{i=1}^{m} r_i\left(\frac{{{\bf x}^o}}{\|{\bf x}^o\|}-\frac{{\bf x}^o-{\bf a}_i}{\|{\bf x}^o-{\bf a}_i\|}\right).
\end{equation*}
Then we obtain the Fisher information matrix
\begin{equation*}
	\begin{split}
		{\bf F}= & \mathbb E\left[\frac{\partial \ell ({\bf d};{\bf x}^o)}{\partial {\bf x}^o} \left( \frac{\partial \ell ({\bf d};{\bf x}^o)}{\partial {\bf x}^o}\right) ^\top \right] \\
		= & \frac{m}{\sigma^2} \frac{{\bf x}^o {{\bf x}^o}^\top}{\|{\bf x}^o\|^2}+\frac{1}{\sigma^2} \sum_{i=1}^{m} \frac{({\bf x}^o-{\bf a}_i)({\bf x}^o-{\bf a}_i)^\top}{\|{\bf x}^o-{\bf a}_i\|^2} \\
		& -\frac{1}{\sigma^2} \sum_{i=1}^{m} \frac{{\bf x}^o({\bf x}^o-{\bf a}_i)^\top+({\bf x}^o-{\bf a}_i){{\bf x}^o}^\top}{\|{\bf x}^o\|\|{\bf x}^o-{\bf a}_i\|}.
	\end{split}
\end{equation*}
The CRLB is given by the trace of the inverse of the Fisher information matrix,  i.e., ${\rm CRLB}={\rm tr}({\bf F}^{-1})$. From~\eqref{asymptotic_normality_LS} we have $\hat {\bf x}^{\rm ML}_{m}$ converges to  the true value ${\bf x}^o$ with the asymptotic covariance of $\frac{\sigma^2}{m} {\bf M}^{-1}({\bf x}^o)$. Further combining the definition of ${\bf M}({\bf x})$, it holds that $\lim\limits_{m \rightarrow \infty} m {\bf F}^{-1}=\sigma^2{\bf M}^{-1}({\bf x}^o)$, which implies that the ML estimator $\hat {\bf x}^{\rm ML}_{m}$ is asymptotically efficient.

Till now, we have given some assumptions on the sensors' coordinates and the measurement noises under which the ML estimate is consistent and asymptotically normal. However, due to the norm functions in the objective, the ML problem~\eqref{ML_problem} is non-convex and hard to solve. In the following section, we will introduce a two-step estimation scheme that can achieve the same asymptotic property as the ML estimate. 

\section{A Two-Step Estimator} \label{Two_step_estimator}
In this section, we introduce a two-step estimation scheme, which can realize the same asymptotic property that the ML estimate possesses, i.e., it achieves the CRLB asymptotically. 
Before that, we show the convergence of the objective function (denoted as $P_m({\bf x})$) of the ML problem~\eqref{ML_problem}.
\begin{theorem} \label{convergence_of_lik}
	Given Assumptions~\ref{assumption_1}-\ref{convergence_of_sample_distribution},~\ref{deployment_of_sensor2}(i-ii), the negative log-likelihood function $P_m({\bf x})$ converges uniformly to $P({\bf x}):=\|f({\bf x}^o)-f({\bf x})\|_t^2 + \sigma^2$ on $\mathcal X$. In addition, $\nabla^2 P({\bf x}^o)=2{\bf M}({\bf x}^o)$.
\end{theorem}
\begin{proof}
	The proof is based on the following lemma. 
	\begin{lemma}[ {\cite[Lemma 4]{zeng2022globally}}] \label{property_of_bounded_variance}
		Let $\{X_k\}$ be a sequence of independent random variables with $\mathbb E[X_k]=0$ and $\mathbb E\left[{X_k}^2 \right]  \leq \varphi <\infty$ for all $k$. Then, there holds $\sum_{k=1}^{m}X_k/\sqrt{m}=O_p(1)$.
	\end{lemma}
	Note that $d_i=\|{\bf a}_i-{\bf x}^o\|-\|{\bf x}^o\|+r_i$. Let $r=(r_i)$. We have
	\begin{align*}
		P_m({\bf x}) & = \frac{1}{m}\sum_{i=1}^{m} (d_i+\|{\bf x}\|-\|{\bf a}_i-{\bf x}\|)^2 \\
		& = \frac{1}{m}\sum_{i=1}^{m} (f_i({\bf x}^o)-f_i({\bf x})+r_i)^2 \\
		& \rightarrow \|f({\bf x}^o)-f({\bf x})\|_t^2 + \left\langle f({\bf x}^o)-f({\bf x}),r \right\rangle_t + \|r\|_t^2 \\
		& = P({\bf x}),
	\end{align*}
	uniformly for ${\bf x} \in \mathcal X$, where $\left\langle f({\bf x}^o)-f({\bf x}),r \right\rangle_t=0$ according to Lemma~\ref{property_of_bounded_variance}. Since $\|f({\bf x}^o)-f({\bf x})\|_t^2$ has a unique minimum at ${\bf x}={\bf x}^o$ (Lemma~\ref{Lemma_asymptotically_localizable}), so does $P({\bf x})$. 
	
	Next, we show the positive definiteness of $\nabla^2 P({\bf x}^o)$. Note that $\nabla P_m(\hat {\bf x}^{\rm ML}_m)=0$. The Taylor expansion of $P_m({\bf x})$ in a small neighborhood of $\hat {\bf x}^{\rm ML}_m$ is
	\begin{equation*}
		P_m({\bf x})=P_m(\hat {\bf x}^{\rm ML}_m)+\frac{1}{2} \|{\bf x}-\hat {\bf x}^{\rm ML}_m\|_{\nabla^2 P_m(\hat {\bf x}^{\rm ML}_m)}^2  + o(\|{\bf x}-\hat {\bf x}^{\rm ML}_m\|^2),
	\end{equation*}
	where $\nabla^2 P_m({\bf x})=\frac{1}{m} \sum_{i=1}^{m} \nabla f_i({\bf x})\nabla f_i({\bf x})^\top-\nabla^2 f_i({\bf x})(d_i-f_i({\bf x}))$. 
	Note that
	\begin{align*}
		\nabla^2 P_m({\bf x}^o) &= \frac{2}{m} \sum_{i=1}^{m} \nabla f_i({\bf x}^o)\nabla f_i({\bf x}^o)^\top-\nabla^2 f_i({\bf x}^o)r_i \\
		& \rightarrow 2 {\bf M}({\bf x}^o). 
	\end{align*}
	Since $P_m({\bf x})$ converges uniformly to $P({\bf x})$ on $\mathcal X$, and $\nabla^2 P_m({\bf x})$ convergences uniformly on $\mathcal X$ (by combining Assumptions~\ref{compact_and_bounded},\ref{convergence_of_sample_distribution}, and the Helly-Bray theorem), we have $\nabla^2 P_m({\bf x}^o) \rightarrow \nabla^2 P({\bf x}^o)$, which implies $\nabla^2 P({\bf x}^o)=2 {\bf M}({\bf x}^o)$. From Lemma~\ref{nonsingular_fisher} we know that ${\bf M}({\bf x}^o)$ is positive definite. Hence, $\nabla^2 P({\bf x}^o)$ is positive definite. %Therefore, in the asymptotic case, $P_m({\bf x})$ is convex in a  neighborhood of $\hat {\bf x}^{\rm ML}_m$. In other words, $P({\bf x})$ can be approximated by a convex function in a neighborhood around ${\bf x}^o$.
\end{proof}
An example of $P_m({\bf x})$ is depicted in Fig.~\ref{contour}, from which we can see the convergence of $P_m({\bf x})$. From Theorem~\ref{convergence_of_lik} we know that although $P({\bf x})$ is non-convex, it can be approximated by a convex function in a neighborhood around the global minimum ${\bf x}^o$. When using local iterative methods, e.g., GN iterations, if the initial value falls within this neighboring attraction region, the global minimum can be found. Then the question is how to obtain a desirable initial value. A solution is to devise a consistent estimator, ensuring that with the increase of measurement number, the estimate can converge to the true value to which the ML solution also converges. Formally, the two-step estimation scheme is as follows~\cite{mu2017globally}:

\begin{figure*}[!t]
	\centering
	\begin{subfigure}[b]{.44\textwidth}
		\centering
		\includegraphics[width=\textwidth]{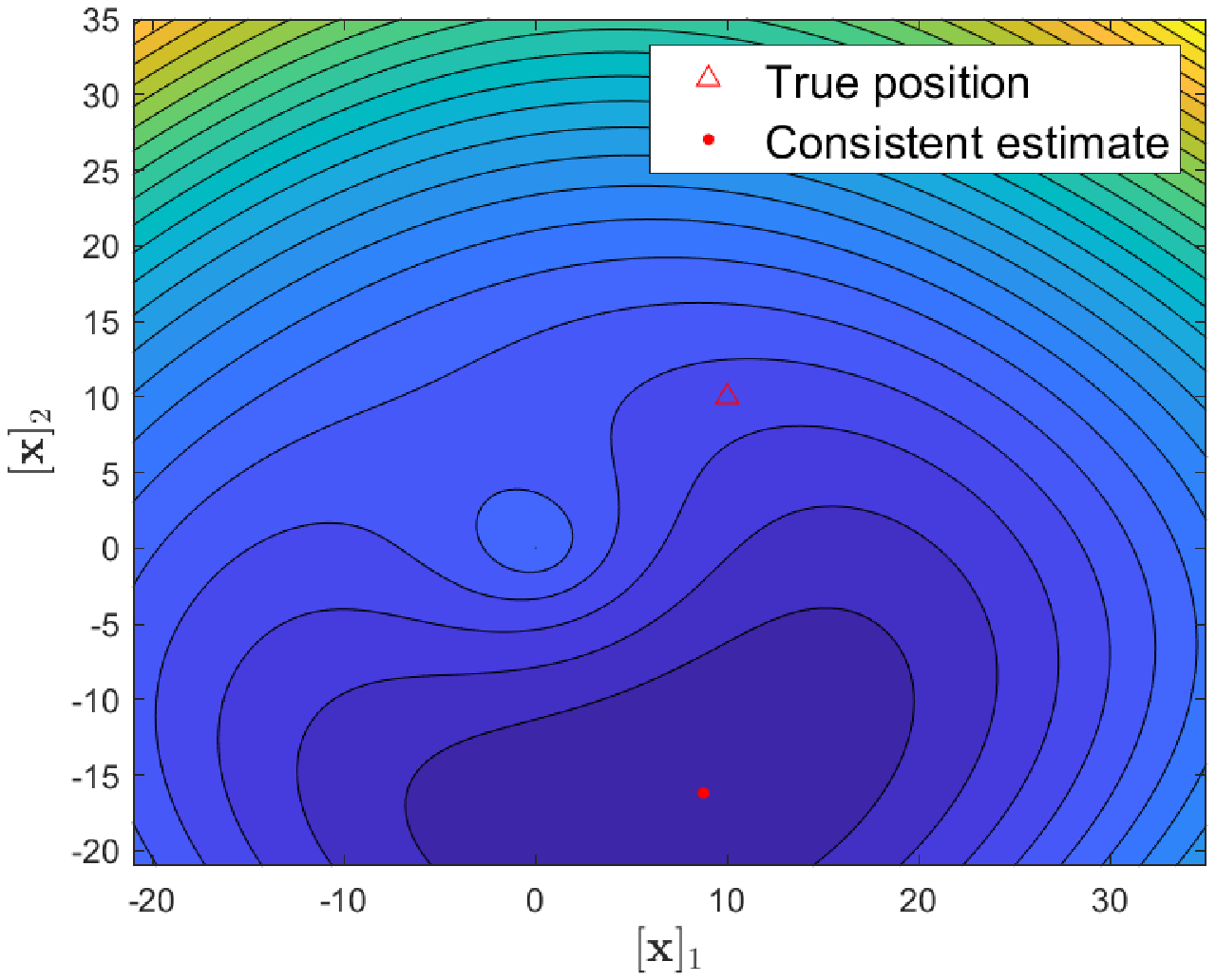}
		\caption{$m=4$}
		\label{t=1}
	\end{subfigure}
	\begin{subfigure}[b]{.44\textwidth}
		\centering
		\includegraphics[width=\textwidth]{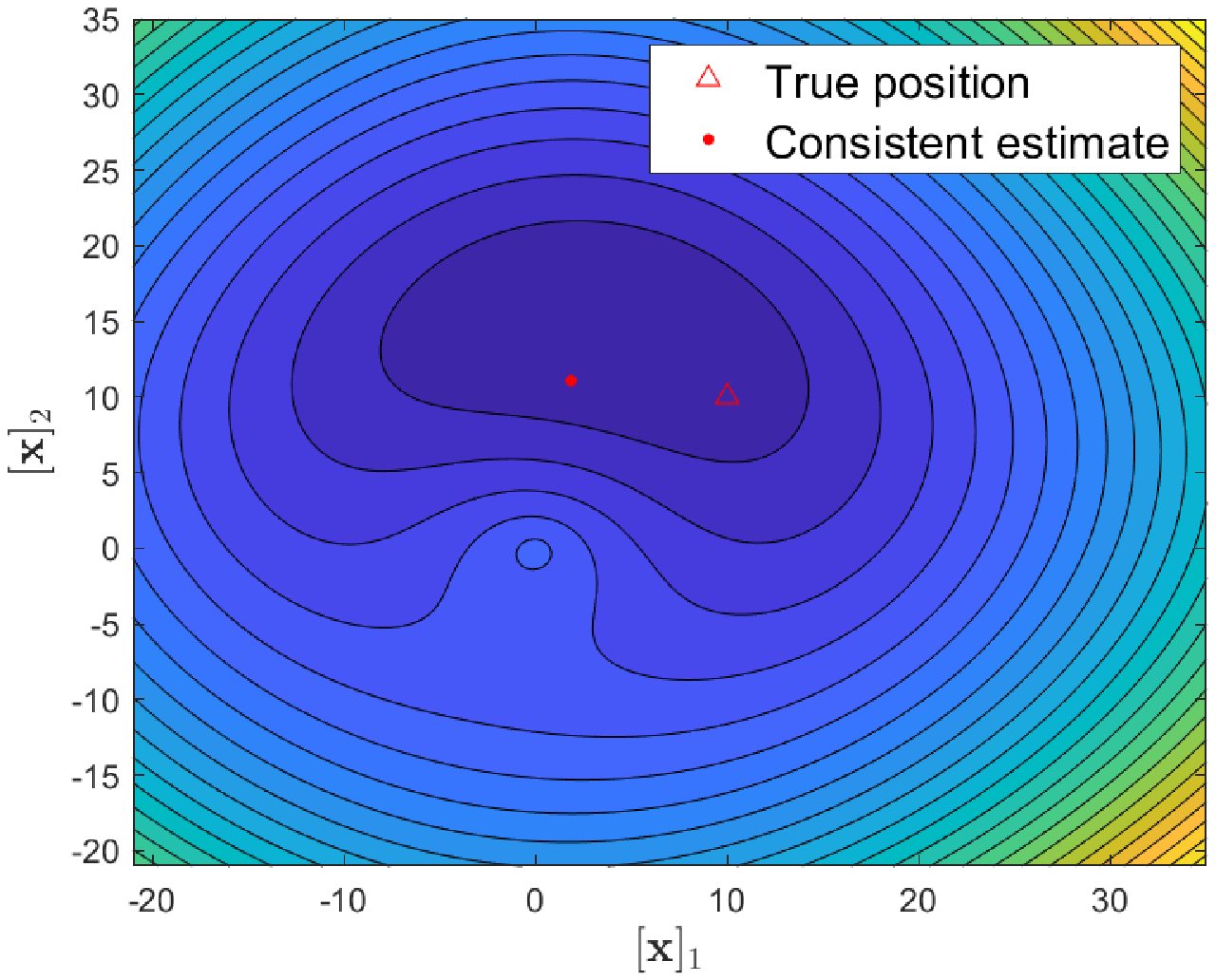}
		\caption{$m=40$}
		\label{t=10}
	\end{subfigure}
	\begin{subfigure}[b]{.44\textwidth}
		\centering
		\includegraphics[width=\textwidth]{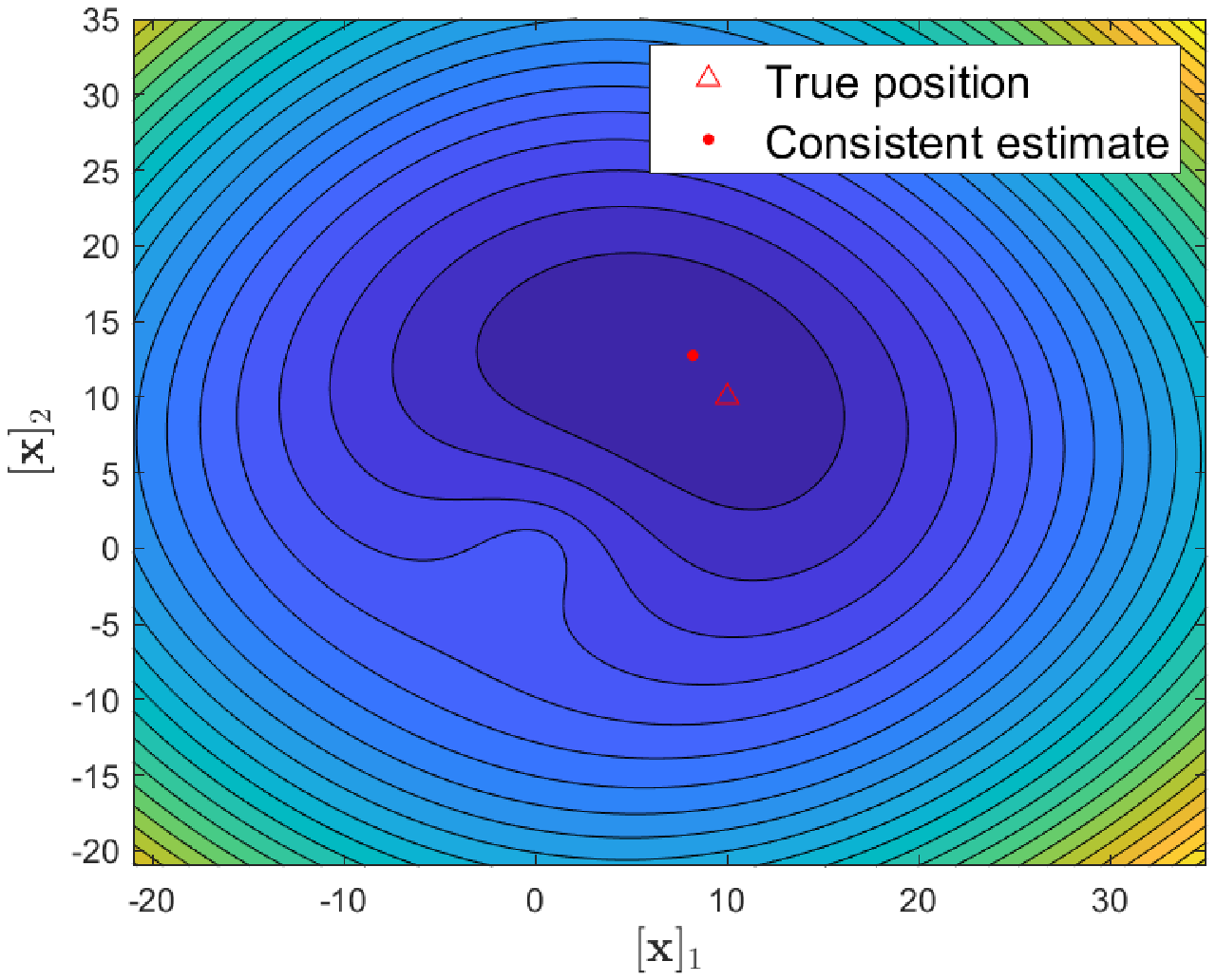}
		\caption{$m=400$}
		\label{t=100}
	\end{subfigure}
	\begin{subfigure}[b]{.44\textwidth}
		\centering
		\includegraphics[width=\textwidth]{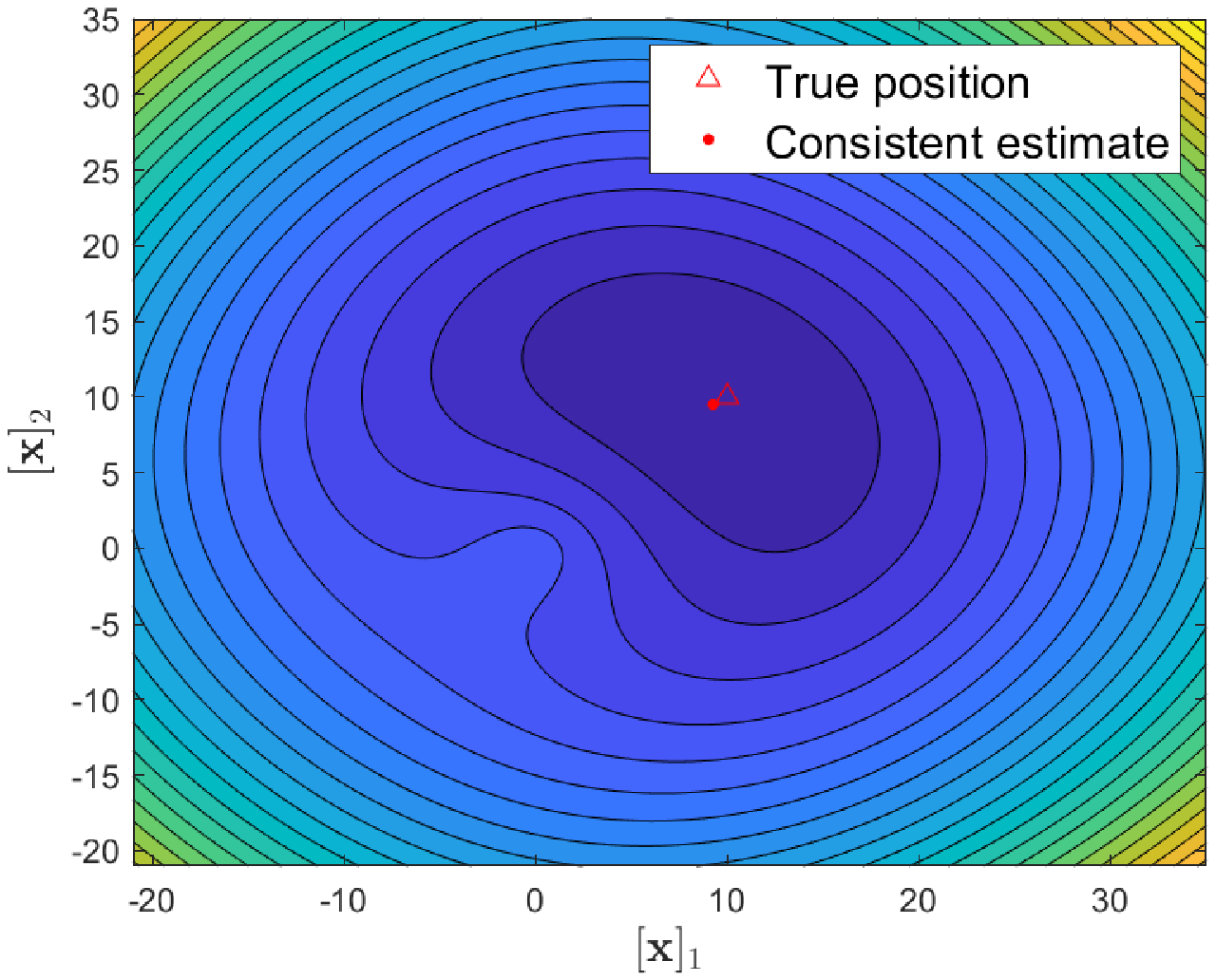}
		\caption{$m=4000$}
		\label{t=1000}
	\end{subfigure}
	\caption{Contours of the objective function~\eqref{ML_problem} and consistent estimates in the 2D case.}
	\label{contour}
\end{figure*}

\noindent {\bf Step 1.} Determine a globally consistent estimate for the source coordinates.

\noindent {\bf Step 2.} Use this preliminary estimate as an initial value for some algorithms that determine the ML estimator.

In Step 2, the GN algorithm is often used to improve the accuracy of the consistent estimate obtained in Step 1. 
The GN iterations associated with the ML problem~\eqref{ML_problem} have the following iterative form:
\begin{equation}\label{GN_iteration}
	\hat {\bf x}_m(k+1)=\hat {\bf x}_m(k)+\left( {\bf J}^\top (k) {\bf J}(k)\right)^{-1}  {\bf J}^\top (k)({\bf d}- {\bf f}_k),
\end{equation} 
where
\begin{align*}
	{\bf f}_k&=\left[ f_1(\hat {\bf x}_m(k)),\ldots,f_m(\hat {\bf x}_m(k))\right] ^\top, \\
	{\bf J}(k)&=\left[\nabla f_1(\hat {\bf x}_m(k)),\ldots,\nabla f_m(\hat {\bf x}_m(k)) \right]^\top.
\end{align*}

The two-step scheme described above has the attractive property that only a one-step GN iteration is sufficient to theoretically guarantee that the resulting two-step estimate has the same asymptotic property as the ML solution~\cite{lehmann2006theory}. The conclusion is summarized in the following result. 
\begin{lemma}[ {\cite[Theorem 4.3]{lehmann2006theory}}] \label{theorem_two_step}
	Suppose that $\hat {\bf x}_{m}$ is a $\sqrt{m}$-consistent estimate of ${\bf x}^o$, i.e., $\hat {\bf x}_{m}-{\bf x}^o=O_p(1/\sqrt{m})$. Denote the one-step GN iteration of $\hat {\bf x}_{m}$ by $\hat {\bf x}^{\rm GN}_m$. Then, under Assumptions~\ref{assumption_1}-\ref{convergence_of_sample_distribution},~\ref{deployment_of_sensor2}(i-ii), we have 
	\begin{equation*}
		\hat {\bf x}^{\rm GN}_m-\hat {\bf x}^{\rm ML}_{m}=o_p(1/\sqrt{m}).
	\end{equation*}
	That is $\hat {\bf x}^{\rm GN}_m$ has the same asymptotic property that $\hat {\bf x}^{\rm ML}_{m}$ possesses. 
\end{lemma}
The same asymptotic property means that $\hat {\bf x}^{\rm GN}_m$ is also consistent and asymptotically efficient. From Lemma~\ref{theorem_two_step} we know that the key in the two-step estimation scheme is to obtain a $\sqrt{m}$-consistent estimate in the first step. In the subsequent section, we will propose such an estimator. 

\section{Bias Elimination and Closed-Form Consistent Localization Method}
In this section, we will devise a $\sqrt{m}$-consistent estimator. First, we relax the original non-convex problem into a linear least-squares problem. We then analyze and derive the bias of the resulting least-squares solution by virtue of its closed-form expression. By eliminating the bias, we obtain an asymptotically unbiased (thus consistent) estimate. 

Specifically, by moving $\|{\bf x}^o\|$ in~\eqref{range_difference_measurements} to the left side and then squaring both sides, we obtain the modified model:
\begin{equation}\label{squared_range_difference}
	d_i^2-\|{\bf a}_i\|^2=-2d_i\|{\bf x}^o\|-2{\bf a}_i^\top {\bf x}^o+e_i, 
\end{equation}
where $e_i=2\|{\bf a}_i-{\bf x}^o\|r_i+r_i^2$. Since all items related to $r_i$ are contained in $e_i$, $e_i$ characterizes the noise in~\eqref{squared_range_difference}. For model~\eqref{squared_range_difference}, on the one hand, the mean of the noise term $e_i$ is $\sigma^2$, not $0$; on the other hand, the regressor term $[-2d_i,-2{\bf a}_i^\top]$ is correlated with the noise term $e_i$. As a result, the least-squares solution based on~\eqref{squared_range_difference} is biased and thus not consistent. With the prior knowledge of the measurement noises' variance $\sigma^2$, we can subtract it from both sides of~\eqref{squared_range_difference}, yielding
\begin{equation}\label{subtrct_variance}
	d_i^2-\|{\bf a}_i\|^2-\sigma^2=-2d_i\|{\bf x}^o\|-2{\bf a}_i^\top {\bf x}^o+\epsilon_i, 
\end{equation}
where $\epsilon_i=e_i-\sigma^2$ has zero mean. Let ${\bf y}^o=[{{\bf x}^o}^\top, \|{\bf x}^o\|]^\top$. By stacking~\eqref{subtrct_variance} for $m$ sensors, we obtain the following matrix form:
\begin{equation}\label{subtrct_variance_matrix_form}
	{\bf b}={\bf A}{\bf y}^o+{\bm \epsilon},
\end{equation}
where 
\begin{equation*}
	{\bf A}=\begin{bmatrix}
		-2{\bf a}_1^\top & -2d_1  \\
		\vdots      & \vdots\\
		-2{\bf a}_m^\top & -2d_m
	\end{bmatrix},
	{\bf b}=\begin{bmatrix}
		d_1^2-\|{\bf a}_1\|^2-\sigma^2   \\
		\vdots      \\
		d_m^2-\|{\bf a}_m\|^2-\sigma^2
	\end{bmatrix},
	{\bm \epsilon}=\begin{bmatrix}
		\epsilon_1 \\
		\vdots \\
		\epsilon_m
	\end{bmatrix}.
\end{equation*}
Since the sensors are not collinear (coplanar) and $d_i$'s contain random noises, the matrix ${\bf A}^\top {\bf A}$ is almost surely invertible. Then the closed-form solution can be obtained
\begin{equation} \label{biased_closed_form}
	\hat {\bf y}_m^{\rm B}=( {\bf A}^\top {\bf A} )^{-1} {\bf A}^\top {\bf b}.  
\end{equation}
Since the regressor $\bf A$ which contains $d_i$ is correlated with the noise term $\bm \epsilon$, $\hat {\bf y}_m^{\rm B}$ is biased. In what follows, we will eliminate the bias of $\hat {\bf y}_m^{\rm B}$ and obtain a consistent estimate. Before that, we rephrase~\eqref{subtrct_variance} as follows:
\begin{equation}\label{subtrct_variance_ideal}
	d_i^2-\|{\bf a}_i\|^2-\sigma^2=-2d_i^o\|{\bf x}^o\|-2{\bf a}_i^\top {\bf x}^o+\eta_i, 
\end{equation}
where $\eta_i=\epsilon_i-2\|{\bf x}^o\|r_i$ and $d_i^o=\|{\bf a}_i-{\bf x}^o\|-\|{\bf x}^o\|$. 
The matrix form of~\eqref{subtrct_variance_ideal} is 
\begin{equation} \label{ideal_matrix_form}
	{\bf b}= {\bf A}^o {\bf y}^o + {\bm \eta},
\end{equation}
where 
\begin{equation*}
	{\bf A}^o=\begin{bmatrix}
		-2{\bf a}_1^\top & -2d_1^o  \\
		\vdots      & \vdots\\
		-2{\bf a}_m^\top & -2d_m^o
	\end{bmatrix},
	{\bm \eta}=\begin{bmatrix}
		\eta_1 \\
		\vdots \\
		\eta_m
	\end{bmatrix}.
\end{equation*}
\begin{lemma} \label{ideal_invertible}
	Given Assumption~\ref{deployment_of_sensor}, the matrix ${\bf A}^{o\top} {\bf A}^o$ is invertible.
\end{lemma}
\begin{proof}
	Suppose ${\bf A}^{o \top} {\bf A}^o$ is singular, i.e., ${\bf A}^o$ is not full column rank. There exists a ${\bf y} \neq 0$ such that ${\bf A}^o {\bf y}=0$, i.e., $-2{\bf a}_i^\top [{\bf y}]_{1:n}=2(\|{\bf x}^o-{\bf a}_i\|-\|{\bf x}^o\|)[{\bf y}]_{n+1}$ for all $i \in\{1,\ldots,m\}$. By some algebraic operations, we obtain
	\begin{equation*} 
	({\bf a}_i^\top [{\bf y}]_{1:n} +2 [{\bf y}]_{n+1}\|{\bf x}^o\|)[{\bf y}]_{1:n}^\top {\bf a}_i+[{\bf y}]_{n+1}^2(2{\bf x}^o-{\bf a}_i)^\top {\bf a}_i=0,
	\end{equation*}
	which implies ${\bf a}_i$ locate on a conic section (when $n=2$) or a quadric surface (when $n=3$), which contradicts Assumption~\ref{deployment_of_sensor} and completes the proof. 
\end{proof}
Then another closed-form solution can be obtained
\begin{equation} \label{unbiased_closed_form}
	\hat {\bf y}_m^{\rm UB}=( {\bf A}^{o \top} {\bf A}^o )^{-1} {\bf A}^{o \top} {\bf b}. 
\end{equation}
For the estimate~\eqref{unbiased_closed_form}, we have the following theorem. 
\begin{theorem} \label{consistency_of_ideal}
	The estimate $\hat {\bf y}_m^{\rm UB}$ is $\sqrt{m}$-consistent, i.e., $\hat {\bf y}_m^{\rm UB} - {\bf y}^o=O_p(1/\sqrt{m})$.
\end{theorem}
\begin{proof} 
	Recall that $\hat {\bf y}^{\rm UB}_m=({\bf A}^{o \top} {\bf A}^o)^{-1} {\bf A}^{o \top} {\bf b}=\left(\frac{1}{m} {\bf A}^{o \top} {\bf A}^o\right)^{-1}\left(\frac{1}{m} {\bf A}^{o \top} {\bf b}\right)  $. For $\frac{1}{m} {\bf A}^{o \top} {\bf b}$ we have 
	\begin{align*}
		\frac{1}{m}{\bf A}^{o \top} {\bf b} =&\frac{1}{m}{\bf A}^{o \top} \begin{bmatrix}
			-2d_1^o\|{\bf x}^o\|-2{\bf a}_1^\top {\bf x}^o+\eta_1  \\
			\vdots \\
			-2d_m^o\|{\bf x}^o\|-2{\bf a}_m^\top {\bf x}^o+\eta_m
		\end{bmatrix}
		\\
		=& \frac{1}{m}{\bf A}^{o \top} {\bf A}^o \begin{bmatrix}
			{\bf x}^o  \\
			\|{\bf x}^o\|
		\end{bmatrix}+\frac{1}{\sqrt{m}} O_p\left(1 \right),
	\end{align*}
	where the second equality is based on Lemma~\ref{property_of_bounded_variance}. Thus, we obtain
	\begin{equation*}
		\sqrt{m}\left( \hat {\bf y}^{\rm UB}_m-{\bf y}^o\right)= \left( \frac{1}{m}{\bf A}^{o \top} {\bf A}^o\right)^{-1}O_p\left(1 \right)=O_p\left(1 \right),
	\end{equation*}
	where the second equality holds because $ \frac{1}{m}{\bf A}^{o \top} {\bf A}^o$ converges to a constant matrix whose elements have the form of tail products. Hence, $\hat {\bf y}^{\rm UB}_m$ converges to ${\bf y}^o$ at a rate of $1/\sqrt{m}$, which completes the proof.
\end{proof}

Although we have proven $\hat {\bf y}^{\rm UB}_m$ is $\sqrt{m}$-consistent, the involved matrix ${\bf A}^o$ is the noise-free counterpart of $\bf A$ and is unknown in practice. Note that the available information is $\bf A$ and $\bf b$. The main idea of bias elimination is to analyze the gap between $\frac{1}{m} {\bf A}^\top {\bf A}$ and $\frac{1}{m} {\bf A}^{o \top} {\bf A}^o$ and that between $\frac{1}{m} {\bf A}^\top {\bf b}$ and $\frac{1}{m} {\bf A}^{o \top} {\bf b}$. By subtracting the gaps, we can eliminate the bias of $\hat {\bf y}_m^{\rm B}$ and achieve the solution $\hat {\bf y}_m^{\rm UB}$ asymptotically. Let
\begin{equation*}
	{\bf G}=\begin{bmatrix}
		{\bf 0}_{1 \times n}~~-2 \\
		~\vdots~~~~~~\vdots \\
		{\bf 0}_{1 \times n}~~-2
	\end{bmatrix}. 
\end{equation*}
We propose the following bias-eliminated estimate
\begin{equation} \label{bias_eliminated_estimate}
	\hat {\bf y}_m^{\rm BE}=\left( \frac{1}{m}{\bf A}^\top {\bf A} - \frac{\sigma^2}{m} {\bf G}^\top {\bf G} \right)^{-1} \left( \frac{1}{m}{\bf A}^\top {\bf b} - \frac{2 \sigma^2}{m} {\bf G}^\top {\bf d}\right).
\end{equation}

\begin{theorem} \label{consistency_of_bias_eliminated}
	The bias-eliminated estimate $\hat {\bf y}_m^{\rm BE}$ is $\sqrt{m}$-consistent, i.e., $\hat {\bf y}_m^{\rm BE} - {\bf y}^o=O_p(1/\sqrt{m})$.
\end{theorem}
\begin{proof}
	Let 
\begin{equation*}
		\Delta {\bf A}={\bf A}-{\bf A}^o=\begin{bmatrix}
			{\bf 0}_{1 \times n}~~-2r_1 \\
			\vdots~~~~~~~~\vdots \\
			{\bf 0}_{1 \times n}~~-2r_m
		\end{bmatrix}. 
	\end{equation*}
	Then we can decompose $\frac{1}{m}{\bf A}^\top {\bf A}$ as 
	\begin{align*}
		\frac{1}{m}{\bf A}^\top {\bf A} & = \frac{1}{m}({\bf A}^o+\Delta {\bf A})^\top ({\bf A}^o+\Delta {\bf A}) \\
		& = \frac{1}{m}{\bf A}^{o \top} {\bf A}^o + \frac{1}{m}\Delta {\bf A}^\top \Delta {\bf A} + O_p\left(\frac{1}{\sqrt{m}}\right) \\
		& = \frac{1}{m}{\bf A}^{o \top} {\bf A}^o + \frac{\sigma^2}{m}{\bf G}^\top {\bf G} + O_p\left(\frac{1}{\sqrt{m}}\right),
	\end{align*}
	where the second and third equalities are based on Lemma~\ref{property_of_bounded_variance}. Similarly, $\frac{1}{m}{\bf A}^\top {\bf b}$ can be decomposed as 
	\begin{align*}
		\frac{1}{m}{\bf A}^\top {\bf b} & = \frac{1}{m}({\bf A}^o+\Delta {\bf A})^\top {\bf b} \\
		& = \frac{1}{m}{\bf A}^{o \top} {\bf b} + \frac{1}{m}\Delta {\bf A}^\top  \begin{bmatrix}
			-2d_1^o\|{\bf x}^o\|-2{\bf a}_1^\top {\bf x}^o+\eta_1  \\
			\vdots \\
			-2d_m^o\|{\bf x}^o\|-2{\bf a}_m^\top {\bf x}^o+\eta_m
		\end{bmatrix}\\
		& = \frac{1}{m}{\bf A}^{o \top} {\bf b} + \frac{1}{m}\Delta {\bf A}^\top \begin{bmatrix}
			2d_1^o r_1  \\
			\vdots \\
			2d_m^o r_m
		\end{bmatrix} + O_p\left(\frac{1}{\sqrt{m}}\right) \\
		& = \frac{1}{m}{\bf A}^{o \top} {\bf b} + \frac{2 \sigma^2}{m}{\bf G}^\top {\bf d}^o + O_p\left(\frac{1}{\sqrt{m}}\right) \\
		& = \frac{1}{m}{\bf A}^{o \top} {\bf b} + \frac{2 \sigma^2}{m}{\bf G}^\top {\bf d} + O_p\left(\frac{1}{\sqrt{m}}\right),
	\end{align*}
	where ${\bf d}^o=[d_1^o,\cdots,d_m^o]^\top$, and the third, fourth, and fifth equalities are based on Lemma~\ref{property_of_bounded_variance} and the fact that $\mathbb E[r_i^3]=0$. Combing the above decomposition, we have 
\begin{align*}
	\hat {\bf y}_m^{\rm BE} & =\left( \frac{1}{m}{\bf A}^\top {\bf A} - \frac{\sigma^2}{m} {\bf G}^\top {\bf G} \right)^{-1} \left( \frac{1}{m}{\bf A}^\top {\bf b} - \frac{2 \sigma^2}{m} {\bf G}^\top {\bf d}\right) \\
	& = \left( \frac{1}{m}{\bf A}^{o \top} {\bf A}^o + O_p\left(\frac{1}{\sqrt{m}} \right) \right)^{-1} \left( \frac{1}{m}{\bf A}^{o \top} {\bf b} +O_p\left(\frac{1}{\sqrt{m}} \right) \right) \\
	& = \hat {\bf y}_m^{\rm UB} + O_p\left(\frac{1}{\sqrt{m}} \right). 
	\end{align*}
	Since $\hat {\bf y}_m^{\rm UB}$ is $\sqrt{m}$-consistent, so is $\hat {\bf y}_m^{\rm BE}$, which completes the proof. 
\end{proof}

Note that in the bias-eliminated solution~\eqref{bias_eliminated_estimate}, we require prior knowledge of the noise variance. When the noise variance $\sigma^2$ is not available, we need to use an estimated value. The following corollary is an extension of Theorem~\ref{consistency_of_bias_eliminated}.
\begin{corollary}\label{estimated_noise}
	The following bias-eliminated estimate $\hat {\bf y}_m^{\rm BE}$ is still $\sqrt{m}$-consistent if $\hat \sigma^2_m$ is a $\sqrt{m}$-consistent estimate for $\sigma^2$:
\begin{equation} \label{A_bias_eliminated_estimate}
	\hat {\bf y}_m^{\rm BE}=\left( \frac{1}{m}{\bf A}^\top {\bf A} - \frac{\hat \sigma^2_m}{m} {\bf G}^\top {\bf G} \right)^{-1} \left( \frac{1}{m}{\bf A}^\top {\bf b}(\hat \sigma^2_m) - \frac{2 \hat \sigma^2_m}{m} {\bf G}^\top {\bf d}\right),
	\end{equation}
	where ${\bf b}(\hat \sigma^2_m)$ is the vector obtained by replacing $\sigma^2$ with $\hat \sigma^2_m$ in ${\bf b}$.
\end{corollary}
The proof is the same as that of Theorem~\ref{consistency_of_bias_eliminated} by utilizing the facts that $\frac{\hat \sigma^2_m}{m} {\bf G}^\top {\bf G}-\frac{\sigma^2}{m} {\bf G}^\top {\bf G}=O_p(1/{\sqrt{m}})$, $\frac{\hat \sigma^2_m}{m} {\bf G}^\top {\bf d}-\frac{\sigma^2}{m} {\bf G}^\top {\bf d}=O_p(1/{\sqrt{m}})$, and $\frac{1}{m}{\bf A}^\top {\bf b}(\hat \sigma^2_m)-\frac{1}{m}A^\top {\bf b}=O_p(1/{\sqrt{m}})$.

\section{Consistent Noise Variance Estimation}
According to Corollary~\ref{estimated_noise}, if a $\sqrt{m}$-consistent estimate of noise variance $\hat \sigma^2_m$ is available, we can use this estimate to conduct the bias elimination. In this section, we will devise such a $\hat \sigma^2_m$. Define
\begin{equation} \label{tilde_A}
	\tilde {\bf A}=\begin{bmatrix}
		-2{\bf a}_1^\top & 1 & -2d_1  & d_1^2-\|{\bf a}_1\|^2\\
		\vdots      & \vdots & \vdots & \vdots\\
		-2{\bf a}_m^\top & 1 & -2d_m & d_m^2-\|{\bf a}_m\|^2
	\end{bmatrix}, 
\end{equation}
and 
\begin{equation} \label{S_z}
\begin{split}
	{\bf S}(z) & =\begin{bmatrix}
{\bf 0}_{(n+1) \times (n+1)} & {\bf 0}_{(n+1) \times 2} \\
{\bf 0}_{2 \times (n+1)}  & {\bf S}_{22}(z)
\end{bmatrix}, \\
{\bf S}_{22}(z) & =\begin{bmatrix}
4z & -4 \bar d z \\
-4 \bar d z & 4 \overbar {d^2} z-2z^2
\end{bmatrix},
\end{split}
\end{equation}
where 
$\bar d=\sum_{i=1}^{m} d_i/m$ and $\overbar {d^2}=\sum_{i=1}^{m} d_i^2/m$. Further, let ${\bf Q}={\tilde {\bf A}}^\top {\tilde {\bf A}}/m$. If the sensors ${\bf a}_i$'s are not collinear (resp. coplanar) in the 2D (resp. 3D) case, the first $n+1$ columns of $\tilde {\bf A}$ are linearly independent. Further, note that the last two columns of $\tilde {\bf A}$ consist of random noises. As a result, $\tilde {\bf A}$ has full column rank almost surely if there are at least $n+3$ sensors that are not collinear (resp. coplanar) in the 2D (resp. 3D) case. 
Hence, given Assumption~\ref{deployment_of_sensor2}$(ii)$, there exists $m_0>0$ such that the matrix $\bf Q$ is positive definite for any $m \geq m_0$. Since $\bf Q$ is positive definite, the eigenvalues of ${\bf Q}^{-1}{\bf S}(z)$ are all real numbers. Let $\lambda_{\rm max}({\bf Q}^{-1}{\bf S}(z))$ denote the largest eigenvalue of ${\bf Q}^{-1}{\bf S}(z)$. The following theorem gives a $\sqrt{m}$-consistent estimator for the noise variance. 
\begin{theorem} \label{consistent_noise_est}
	The set $\mathcal Z=\{z\in\mathbb R \mid \lambda_{\rm max}({\bf Q}^{-1} {\bf S}(z))=1 \}$ is non-empty. Let $\hat \sigma^2_m = {\rm min} ~\mathcal Z$, then $\hat \sigma^2_m-\sigma^2=O_p(1/\sqrt{m})$.
\end{theorem}
The proof of Theorem~\ref{consistent_noise_est} is presented in Appendix~\ref{proof_consistent_noise_est}.

Now we explicitly illustrate how to calculate the root of $\lambda_{\rm max}( {\bf Q}^{-1}  {\bf S}(z))=1$. We express $\bf Q$ and its inverse as
\begin{equation*}
	{\bf Q} = \begin{bmatrix}
		{\bf Q}_{11} &  {\bf Q}_{12} \\
		{\bf Q}_{21} &  {\bf Q}_{22}
	\end{bmatrix},~~
	{\bf Q}^{-1} = \begin{bmatrix}
		* & * \\
		* & ( {\bf Q}/{\bf Q}_{22})^{-1}
	\end{bmatrix},
\end{equation*}
where $ {\bf Q}_{11} \in \mathbb R^{(n+1) \times (n+1)}$, $ {\bf Q}_{22} \in \mathbb R^{2 \times 2}$, and ${\bf Q}/{\bf Q}_{22}={\bf Q}_{22}- {\bf Q}_{21}{\bf Q}_{11}^{-1} {\bf Q}_{12}$ is the Schur complement of ${\bf Q}_{22}$ of $\bf Q$. Utilizing the structure of $ {\bf S}(z)$, we obtain
\begin{equation*}
	{\rm det}\left(\lambda {\bf I}_{n+3}- {\bf Q}^{-1}  {\bf S}(z)\right)=\lambda^{n+1} {\rm det} \left(\lambda {\bf I}_2-( {\bf Q}/{\bf Q}_{22})^{-1} {\bf S}_{22}(z) \right).
\end{equation*}
Therefore, $\lambda_{\rm max}( {\bf Q}^{-1}  {\bf S}(z))$ is the largest root of the following polynomial with variable $\lambda$:
\begin{equation} \label{implicit_poly}
	{\rm det} \left(\lambda ( {\bf Q}/{\bf Q}_{22}) - {\bf S}_{22}(z) \right)=0.
\end{equation}
Note that $ {\bf Q}/{\bf Q}_{22}$ is symmetric and we express it as
\begin{equation} \label{calculated_matrix}
	{\bf Q}/{\bf Q}_{22} = \begin{bmatrix}
		q_1 & q_2 \\
		q_2 & q_3 
	\end{bmatrix} \in \mathbb R^{2 \times 2}.
\end{equation}
Let $c_1=q_1 q_3-q_2^2$, $c_2=4q_1 \overbar{d^2}+4q_3+8q_2 \bar {d}$, and $c_3=16\overbar{d^2} -16 \bar{d}$. Equation~\eqref{implicit_poly} yields
\begin{equation} \label{explicit_poly}
	c_1 \lambda^2 +(2q_1 z^2-c_2 z) \lambda-8z^3+c_3 z^2=0.
\end{equation}
By setting the largest root of~\eqref{explicit_poly} as $1$, we obtain the following equation:
\begin{equation} \label{larger_root}
	\frac{c_2 z-2q_1 z^2+\sqrt{\Delta}}{2c_1}=1,
\end{equation}
where $\Delta=(2q_1 z^2-c_2 z)^2-4c_1(c_3 z^2-8z^3)$. Equivalently, we can resort to solving the following $3$-order polynomial and choose the solutions that satisfy $2c_1+2q_1 z^2-c_2 z>0$:
\begin{equation} \label{final_poly}
	32c_1 z^3 -(4c_1 c_3+8q_1 c_1)z^2 +4c_1 c_2 z -4 c_1^2 =0.
\end{equation}
Note that the set of the roots of~\eqref{final_poly} that satisfy $2c_1+2q_1 z^2-c_2 z>0$ is $\mathcal Z$. According to Theorem~\ref{consistent_noise_est}, $\mathcal Z$ is non-empty. Finally, the noise variance estimation is set as $\hat \sigma^2_m={\rm min} ~\mathcal Z$. The noise variance estimation procedure is summarized in Algorithm~\ref{pseudo_noise_estimation}. 

The whole procedure of the proposed two-step estimator is summarized in Algorithm~\ref{pseudo_whole_algorithm}. In summary, our proposed estimator includes three procedures---noise variance estimation, bias-eliminated solution calculation~\eqref{A_bias_eliminated_estimate}, and a one-step GN iteration~\eqref{GN_iteration}. Now we analyze the time complexity of our algorithm. In Algorithm~\ref{pseudo_noise_estimation}, note that $\bar d=\sum_{i=1}^{m} d_i/m$, $\overbar {d^2}=\sum_{i=1}^{m} d_i^2/m$, and $\tilde {\bf A} \in \mathbb R^{m \times (n+3)}$. Hence, Line 1 and Line 2 have $O(m)$ (linear) time complexity. For Lines 3, 4, and 5, since ${\bf Q} \in \mathbb R^{(n+3) \times (n+3)}$ and~\eqref{final_poly} is a cubic equation in one variable, they cost $O(1)$ (constant) time. As a result, the whole time complexity of Algorithm~\ref{pseudo_noise_estimation} is $O(m)$. In Algorithm~\ref{pseudo_whole_algorithm}, Line 1 executes Algorithm~\ref{pseudo_noise_estimation}, which has $O(m)$ time complexity. Note that ${\bf A},{\bf G} \in \mathbb R^{m \times (n+1)}$ and ${\bf b},{\bf d} \in \mathbb R^{m \times 1}$. Hence, the time complexity of Line 2 is $O(m)$. Similarly, the time complexity of Line 4 is also $O(m)$. Finally, Line 3 costs $O(1)$ time. Therefore, Algorithm~\ref{pseudo_whole_algorithm} (the whole algorithm) has overall $O(m)$ time complexity, which is appealing in the large sample case.

\begin{algorithm}
	\caption{ Noise variance estimation}
	\label{pseudo_noise_estimation}
	\begin{algorithmic}[1]
		\State Calculate $\bar d$ and $\overbar {d^2}$ in~\eqref{S_z};
		\State Calculate $ {\bf Q}={\tilde {\bf A}}^\top {\tilde {\bf A}}/m$ based on~\eqref{tilde_A};
		\State Calculate $ {\bf Q}/{\bf Q}_{22}$ in~\eqref{calculated_matrix} and $c_1$, $c_2$, $c_3$ in~\eqref{explicit_poly};
		\State Solve~\eqref{final_poly} and collect the roots that satisfy $2c_1+2q_1 z^2-c_2 z>0$ as $\mathcal Z$;
		\State Set $\hat \sigma^2_m={\rm min} ~\mathcal Z$.
	\end{algorithmic}
\end{algorithm}

\begin{algorithm}
	\caption{Consistent and asymptotically efficient estimator}
	\label{pseudo_whole_algorithm}
	\begin{algorithmic}[1]
		\Statex {\bf Input:} sensor coordinates $({\bf a}_i)_{i=1}^{m}$ and range-difference measurements $(d_i)_{i=1}^{m}$.
		\Statex {\bf Output:} source location estimate $\hat {\bf x}^{\rm GN}_m$.
		\State Apply Algorithm~\ref{pseudo_noise_estimation} to obtain an estimate $\hat \sigma^2_m$ of the noise variance $ \sigma^2$;
		\State Calculate the bias-eliminated estimate $\hat {\bf y}^{\rm BE}_m$ according to~\eqref{A_bias_eliminated_estimate};
		\State Set $\hat {\bf x}^{\rm BE}_m=[\hat {\bf y}^{\rm BE}_m]_{1:n}$;
		\State Apply a one-step GN iteration~\eqref{GN_iteration} to obtain $\hat {\bf x}^{\rm GN}_m$.
	\end{algorithmic}
\end{algorithm}

\section{Simulations} \label{simulation}
In this section, we perform simulations to verify our theoretical developments. The algorithms compared with ours (denoted as Bias-Eli) include: (a) CLS: non-convex optimization-based method that optimally solves the spherical least-squares problem~\cite{zeng2022localizability}; (b) GTRS-MPR: modified polar representation solved by generalized trust region subproblem (GTRS) algorithm~\cite{sun2018solution}; (c) BiasRed: bias-reduced solution based on expected bias approximation~\cite{ho2012bias}. At each sensor configuration, the bias and root mean square error (RMSE) of an estimator $\hat {\bf x}$ are approximated as 
\begin{align*}
	\Delta(\hat {\bf x}) &
	=
	\left| \frac{1}{N}\sum_{j=1}^{N}\left( {\hat {\bf x}(\omega_j) - {\bf x}^o}\right)\right|, ~~ {\rm Bias}(\hat {\bf x}) \approx \sum_{i=1}^n [\Delta(\hat {\bf x})]_i,\\
	{\rm RMSE}(\hat {\bf x}) &
	\approx
	\sqrt{\frac{1}{N}\sum_{j=1}^{N}{\left\| {\hat {\bf x}(\omega_j) - {\bf x}^o} \right\|^2}},
\end{align*}
where $N$ is the number of Monte Carlo experiments over measurement noises, and $\hat {\bf x}(\omega_j)$ is the estimate obtained in the $j$-th Monte Carlo test.
For RMSE, we take the root-CRLB (RCRLB) as its baseline.

\subsection{Uniformly distributed sensors}
In this subsection, as in Example~\ref{example_random_vector}, we generate sensors uniformly distributed on the surface of a cube whose center is the origin and edges are $100$m (the reference sensor locates in the origin). The coordinates of the source is ${\bf x}^o=[15,15,15]^\top$, and the standard deviation of Gaussian noises is $\sigma=10$. We set the number of sensors $m$ as $10$, $30$, $100$, $300$, $1000$, and $3000$ respectively, and for each $m$, we run $1000$ Monte Carlo tests to evaluate biases and RMSEs. In each Monte Carlo test, sensor positions and measurement noises are randomly generated.
Now we verify the assumptions proposed in Section~\ref{ML_estimator}. The Gaussian noise assumption (Assumption~\ref{assumption_1}) and the compact and bounded assumption (Assumption~\ref{compact_and_bounded}) are straightforward from this setting. Since the sensors are uniformly generated on the surface of a cube, its sample distribution converges to a uniform distribution over this cube surface whose probability density function is $f_{\mu}({\bf a})=\frac{1}{6 \times 100^2}$ for $\bf a$ on the cube surface, and $f_{\mu}({\bf a})=0$, otherwise. Hence, Assumption~\ref{convergence_of_sample_distribution} holds. Note that the cube surface is not a quadric surface. Therefore, Assumption~\ref{deployment_of_sensor} holds, which implies Assumption~\ref{deployment_of_sensor2}. Therefore, all of the required assumptions have been verified.  

Next, we present the simulation results. From Table~\ref{sigma_estimate} we see that our noise variance estimator is $\sqrt{m}$-consistent. 
The biases under varying $m$ are presented in Fig.~\ref{bias_vs_m}. We see that the proposed Bias-Eli estimator is biased when $m$ is relatively small, while it becomes less biased as $m$ increases. This is because $\frac{\hat \sigma_m^2}{m} {\bf G}^\top {\bf G}$ in~\eqref{bias_eliminated_estimate} characterizes the gap between $\frac{1}{m}{\bf A}^\top {\bf A}$ and $\frac{1}{m}{\bf A}^{o \top} {\bf A}^o$ consistently only in the asymptotic case; it may not be precise when the measurement number is small. So does $\frac{2 \hat \sigma_m^2}{m} {\bf G}^\top {\bf d}$. For the other estimators, their biases converge to nonzero values, which makes them not consistent (as shown in Fig.~\ref{rmse_random_sensor}). 

\begin{table}[htbp]
	\centering
	\caption{RMSE of noise variance estimate. } \label{sigma_estimate}
	\resizebox{0.56\columnwidth}{!}{%
  \begin{tabular}{c c c c c c}
			\hline\hline
			$m=10$ & $m=30$ & $m=100$ & $m=300$ & $m=1000$ & $m=3000$ \\
			\hline
			63.1336  &  30.9802  &  16.5915  &  9.5205  &  5.4593  &  3.0107 \\
			\hline
		\end{tabular}
	}
\end{table}

\begin{figure*}[!t]
	\centering
	\begin{minipage}[b]{.48\textwidth}
		\centering
		\includegraphics[width=0.96\textwidth]{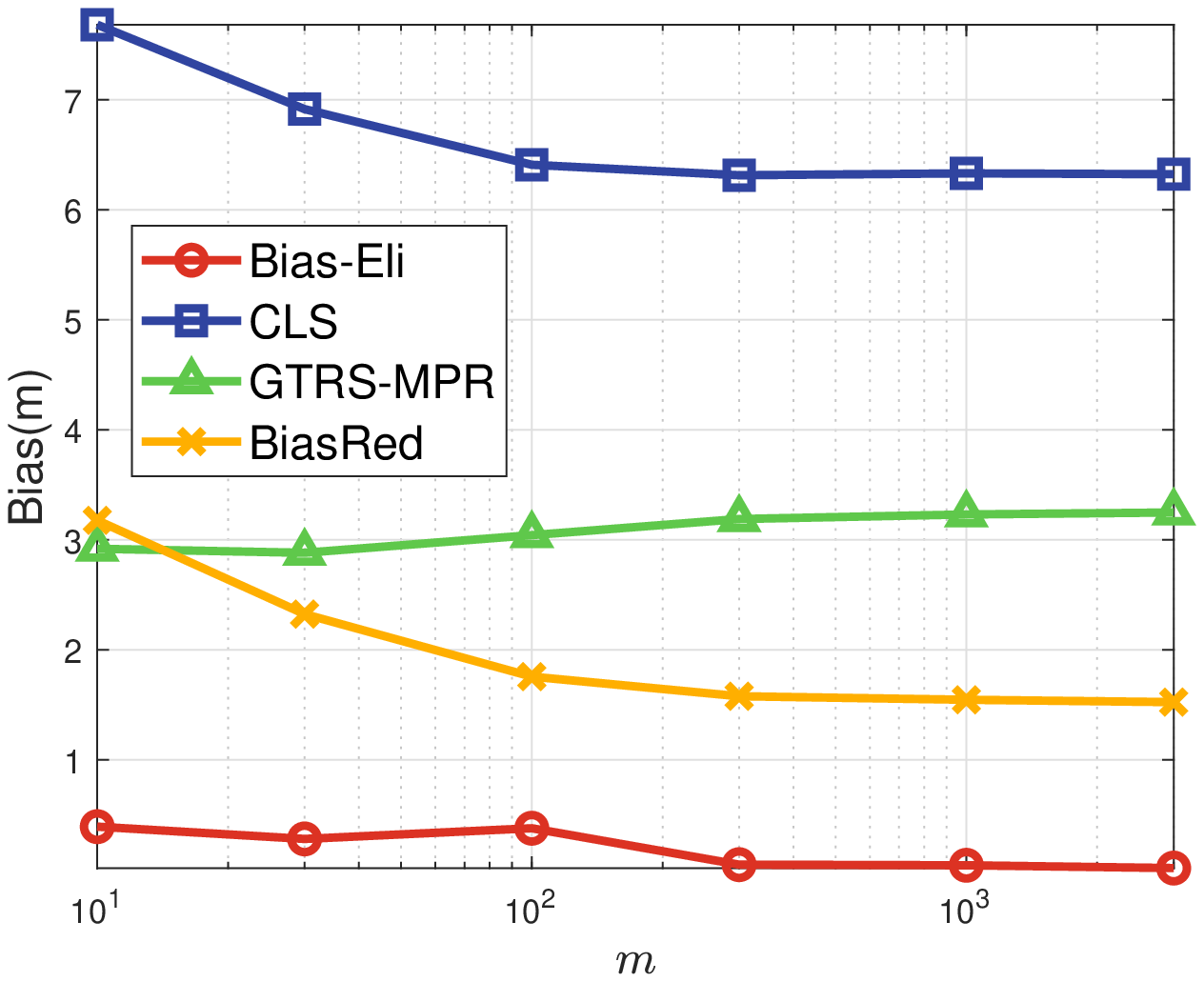}
		\caption{Bias comparison in uniformly distributed case.}
		\label{bias_vs_m}
	\end{minipage}\qquad
	\begin{minipage}[b]{.48\textwidth}
		\centering
		\includegraphics[width=0.96\textwidth]{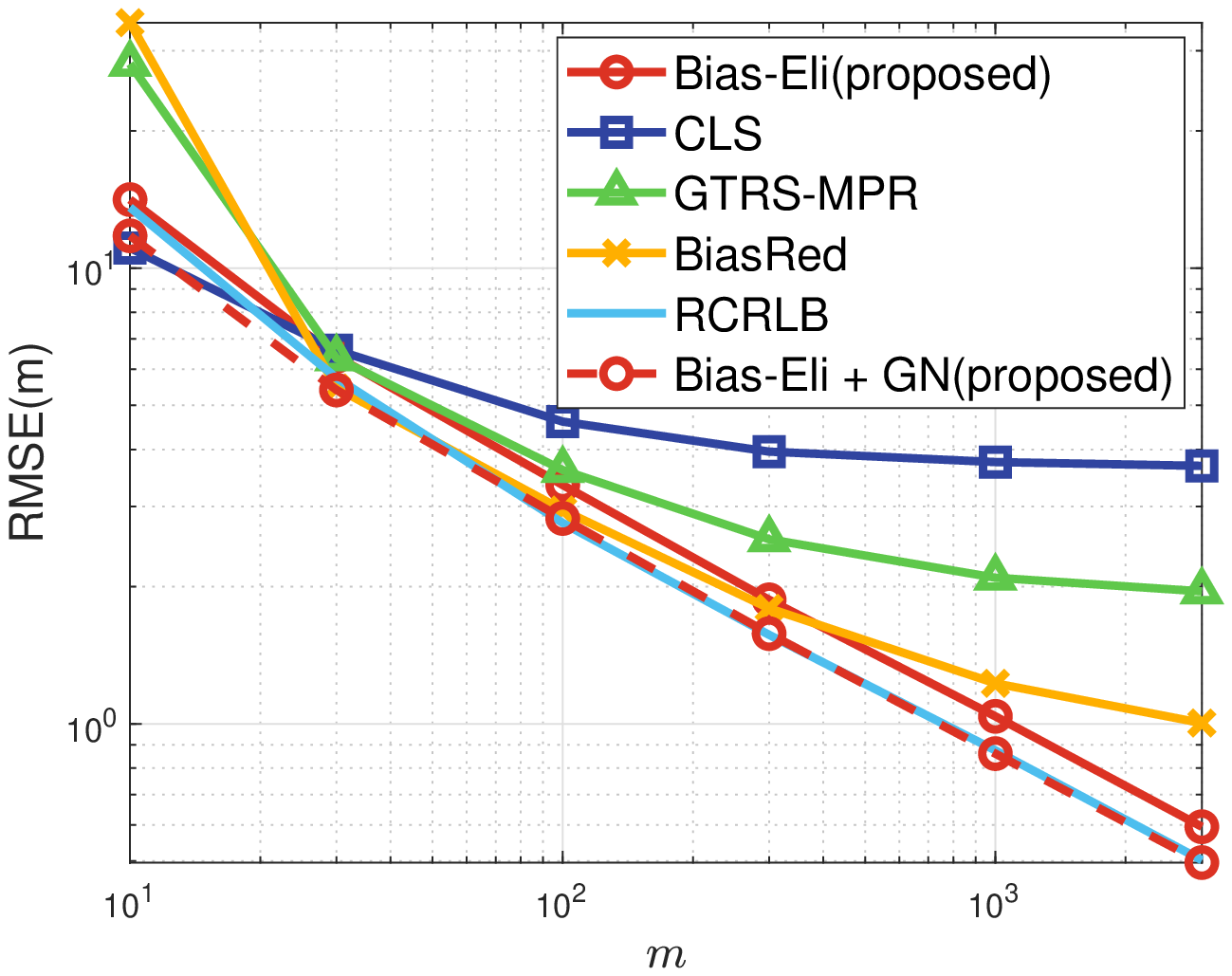}
		\caption{RMSE comparison in uniformly distributed case.}
		\label{rmse_random_sensor}
	\end{minipage}
 \end{figure*}

The RMSEs of the compared estimators under varying $m$ are plotted in Fig.~\ref{rmse_random_sensor}. Since the proposed Bias-Eli solution is asymptotically unbiased, it is consistent, exhibiting a continuously decreasing straight line in the double logarithmic coordinate system. Moreover, the proposed two-step estimator which executes a one-step GN iteration on the basis of the Bias-Eli solution can achieve the RCRLB asymptotically, verifying the claim in Lemma~\ref{theorem_two_step}. 
As has been discussed, the other estimators are not asymptotically unbiased. As a result, they are not consistent, and their RMSEs will be dominated by the asymptotic biases in the large sample case.   

% \begin{figure}[!htbp]
% 	\centering
% 	\includegraphics[width=0.48\textwidth]{./figures/rmse_random_sensor.eps}
% 	\caption{{\color{blue}RMSE comparison in uniformly distributed sensor case.}}
% 	\label{rmse_random_sensor}
% \end{figure}

\subsection{Position-fixed sensors}
In the above subsection, we suppose the sensors are uniformly distributed in a region as Example~\ref{example_random_vector} illustrates. In this subsection, we adopt the scheme in Example~\ref{example_fix_sensors} and fix $10$ sensors at 
\begin{equation}\label{10_fixed_sensors}
    \begin{split}
        	{\bf a}_1&=[50,0,50]^\top,~{\bf a}_2=[50,50,-50]^\top, ~{\bf a}_3=[50,-50,50]^\top, \\
	{\bf a}_4&=[50,0,0]^\top,~{\bf a}_5=[50,50,50]^\top, ~{\bf a}_6=[-50,0,-50]^\top, \\
	{\bf a}_7&=[-50,-50,50]^\top,~{\bf a}_8=[-50,50,-50]^\top, \\
	{\bf a}_9&=[-50,0,0]^\top,~{\bf a}_{10}=[-50,-50,-50]^\top.
    \end{split}
\end{equation}
Each sensor makes total $T$ rounds of i.i.d. observations, and with the increase of $T$, the range-difference measurements can be large enough. The coordinates of the source is ${\bf x}^o=[52,52,52]^\top$, and the standard deviation of noises is $\sigma=5$.     
The measurement number $T$ of each sensor is set to $1$, $3$, $10$, $30$, $100$, and $300$ respectively, and for each $T$, we run $1000$ Monte Carlo tests to evaluate biases and RMSEs. 
Now we verify the assumptions proposed in Section~\ref{ML_estimator}. The Gaussian noise assumption (Assumption~\ref{assumption_1}) and the compact and bounded assumption (Assumption~\ref{compact_and_bounded}) are straightforward from this setting. Note that there are $10$ sensors, and each sensor makes $T$ measurements. It is equivalent to deploying $T$ sensors at each of the $10$ positions and each sensor making one measurement. Then, as $T$ increases, the sample distribution of $10T$ sensors converges to the distribution whose probability measure is $\mu({\bf a})=\frac{1}{10}$ for $\bf a$ that belongs to the set of the $10$ positions in~\eqref{10_fixed_sensors}, and $\mu({\bf a})=0$, otherwise. Hence, Assumption~\ref{convergence_of_sample_distribution} holds. To verify Assumption~\ref{deployment_of_sensor}, we adopt the algebraic method described below this assumption in Section~\ref{ML_estimator}. In this setting, we have $\mathbb E_{\mu} [{\bf v}({\bf a}) {\bf v}({\bf a})^\top]=\sum_{i=1}^{10} {\bf v}({\bf a}_i) {\bf v}({\bf a}_i)^\top/10$, where ${\bf a}_i$'s are listed in~\eqref{10_fixed_sensors}. One can verify that the rank of the matrix $\mathbb E_{\mu} [{\bf v}({\bf a}) {\bf v}({\bf a})^\top]$ is $10$, i.e.,  Assumption~\ref{deployment_of_sensor} holds, which implies that Assumption~\ref{deployment_of_sensor2} also holds. Therefore, all of the required assumptions have been verified. 

Next, we present the simulation results.
The biases under varying $T$ are presented in Fig.~\ref{bias_vs_T}. 
Same as before, our Bias-Eli estimator is asymptotically unbiased.
For the CLS and GTRS-MPR estimators, their biases converge to nonzero values, which makes them not consistent (as shown in Fig.~\ref{rmse}). It is noteworthy that the bias of the BiasRed estimator seems to diverge as $T$ increases. 
The reason may be that a small intensity of noises is assumed in~\cite{ho2012bias}, and the bias-reduced solution is obtained by ignoring the second-order noise terms. In our setting, the noises are not necessarily small enough to apply the result in~\cite{ho2012bias}. The approximation error by ignoring the second-order noise terms may be unstable and increase with $m$. As a result, the bias of the BiasRed estimator increases with $m$.

% \begin{figure}[!htbp]
% 	\centering
% 	\includegraphics[width=0.48\textwidth]{./figures/bias_fixed_sensor.eps}
% 	\caption{{\color{blue}Bias comparison in fixed-sensor case.}}
% 	\label{bias_vs_T}
% \end{figure}

\begin{figure*}[!t]
	\centering
	\begin{minipage}[b]{.48\textwidth}
		\centering
		\includegraphics[width=0.96\textwidth]{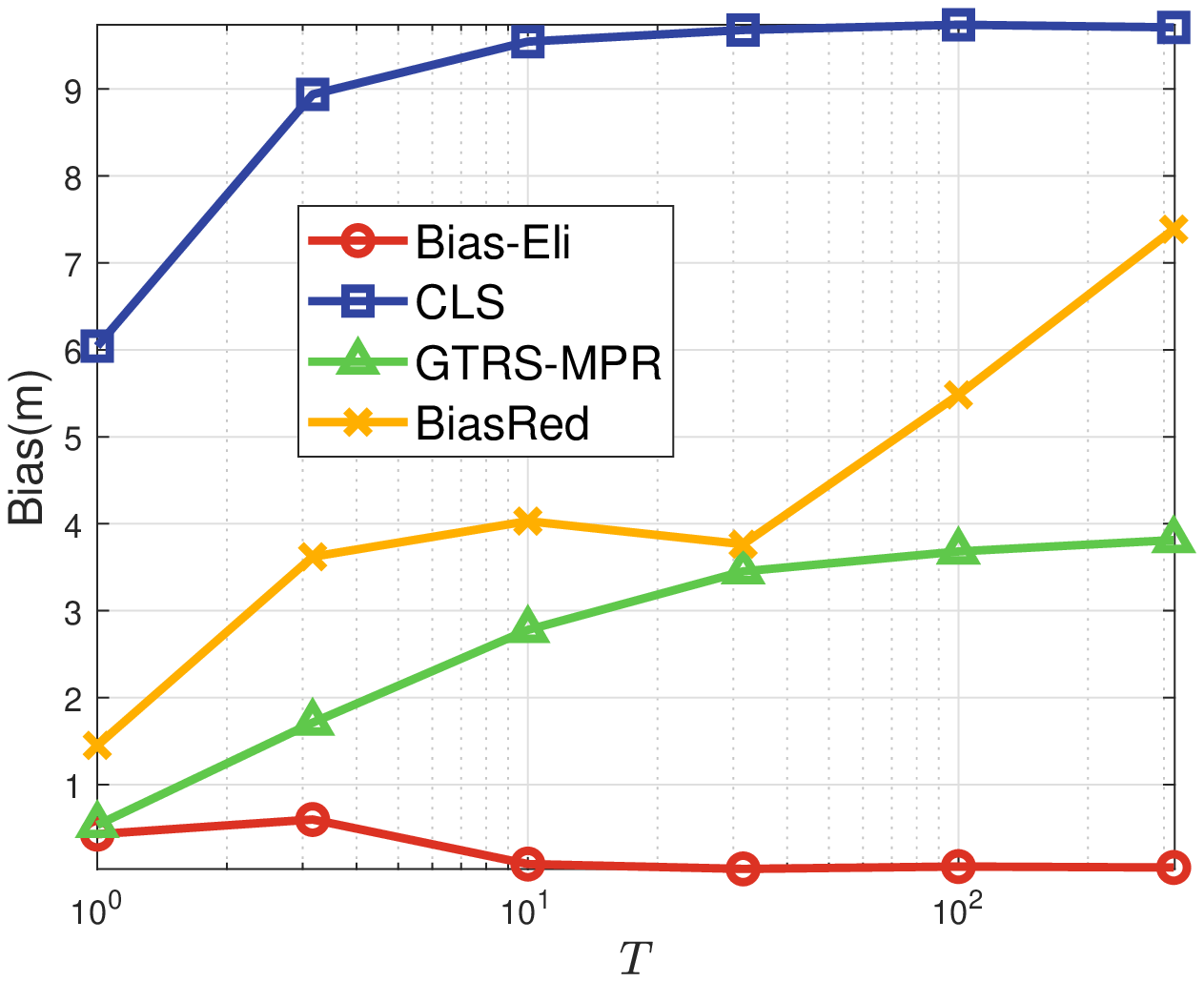}
		\caption{Bias comparison in fixed-sensor case.}
		\label{bias_vs_T}
	\end{minipage}\qquad
	\begin{minipage}[b]{.48\textwidth}
		\centering
		\includegraphics[width=0.96\textwidth]{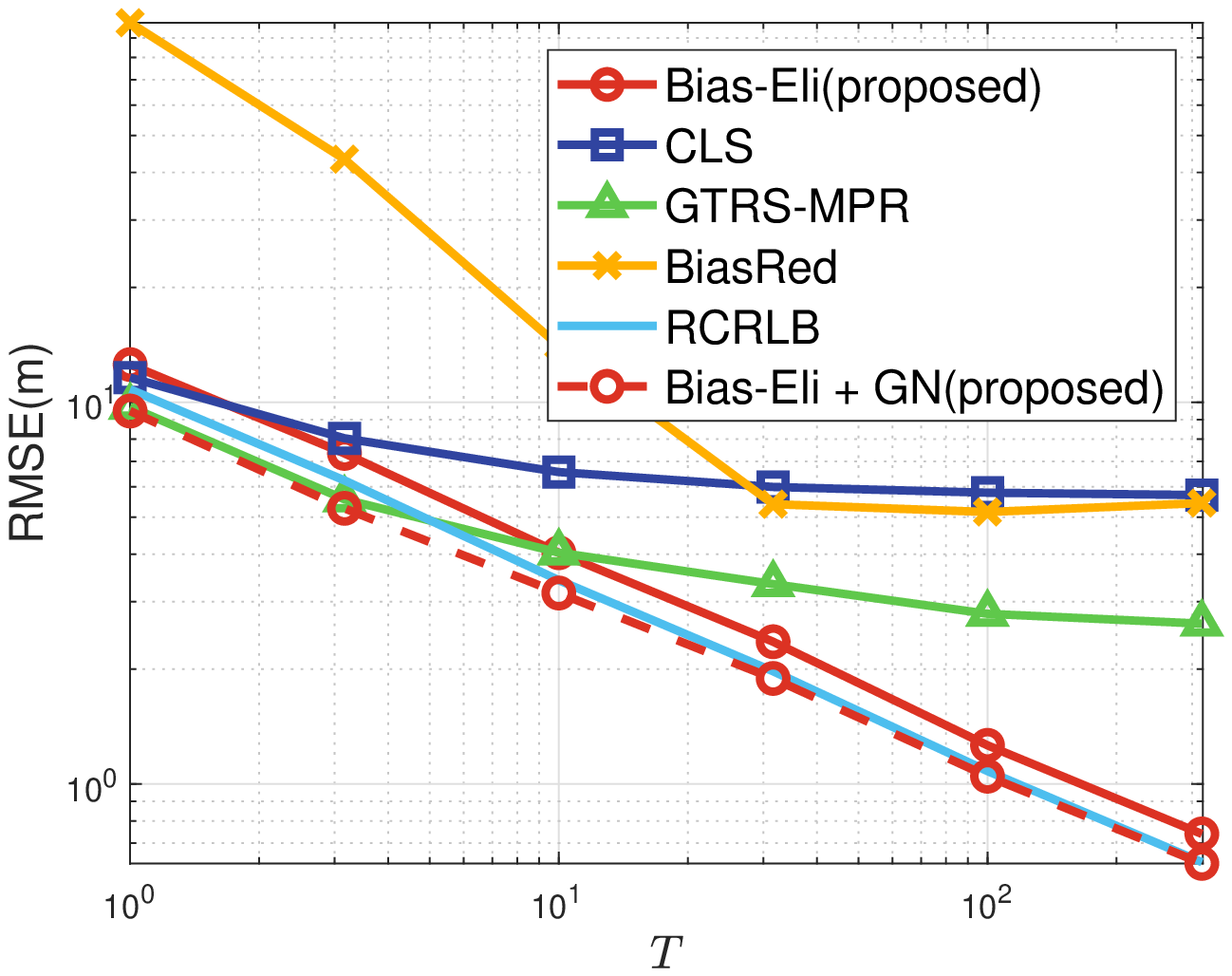}
		\caption{RMSE comparison in fixed-sensor case.}
		\label{rmse}
	\end{minipage}
 \end{figure*}

The RMSEs of the compared estimators under varying $T$ are plotted in Fig.~\ref{rmse}. Our first-step estimate Bias-Eli is consistent, and an additional one-step GN iteration can asymptotically achieve the RCRLB. However, the compared algorithms are not consistent. As discussed above, they are biased even in the asymptotic case, making their RMSEs converge to a nonzero value. It is worth noting that the proposed two-step estimate and the GTRS-MPR solution outperform the RCRLB in the small $T$ region. This is because the estimators are biased in the finite sample case and their variances are not necessarily greater than the CRLB. Besides, we note that the BiasRed estimator, which is based on the assumption of small noise intensity, seems not very stable in this setting, owning a relatively large error.

\begin{figure*}[!t]
	\begin{subfigure}[b]{.48\textwidth}
	\centering
	\includegraphics[width=0.96\textwidth]{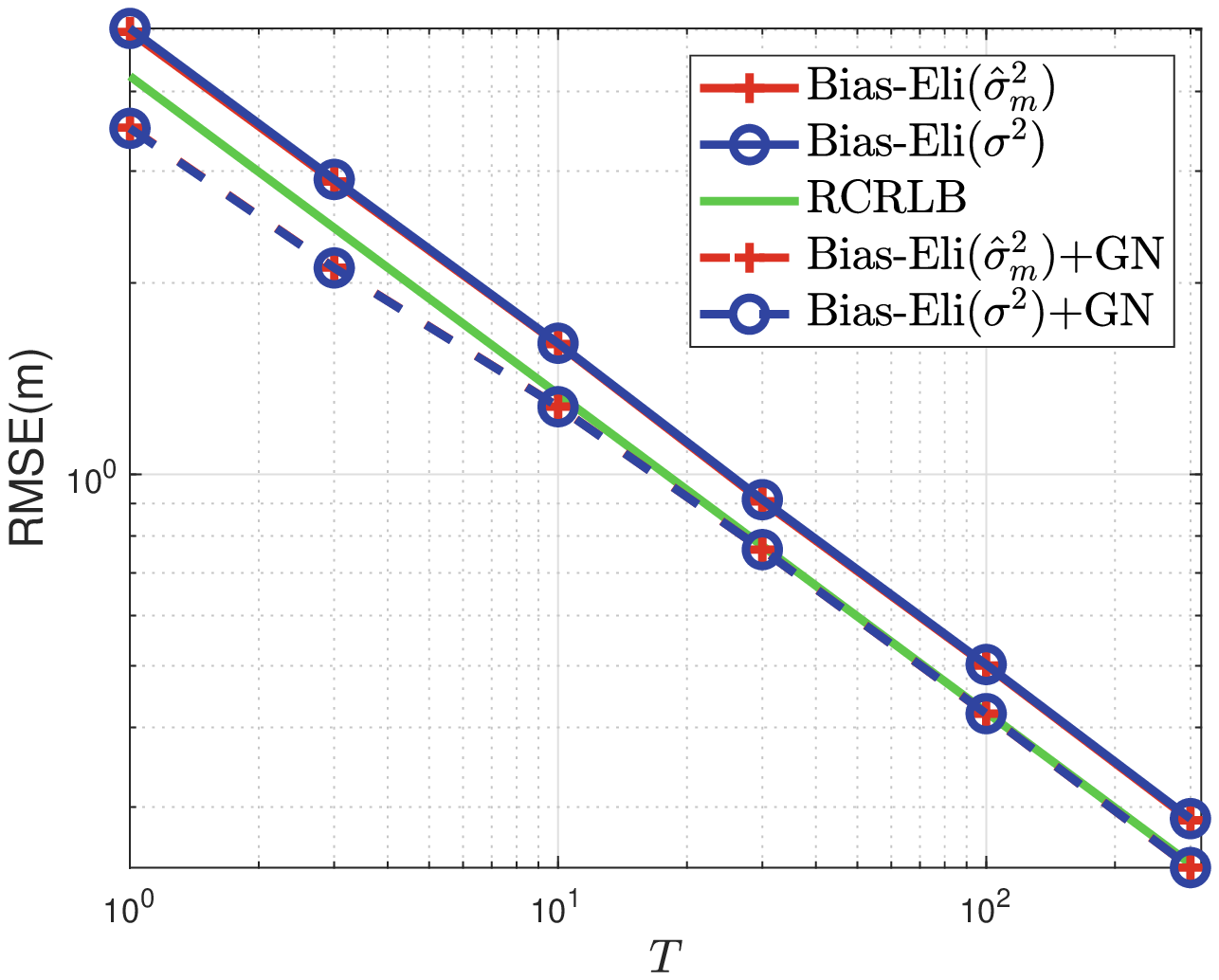}
	\caption{$\sigma=2$}
	\label{NE_vs_BE_varied_m}
\end{subfigure}
\begin{subfigure}[b]{.48\textwidth}
	\centering
	\includegraphics[width=0.96\textwidth]{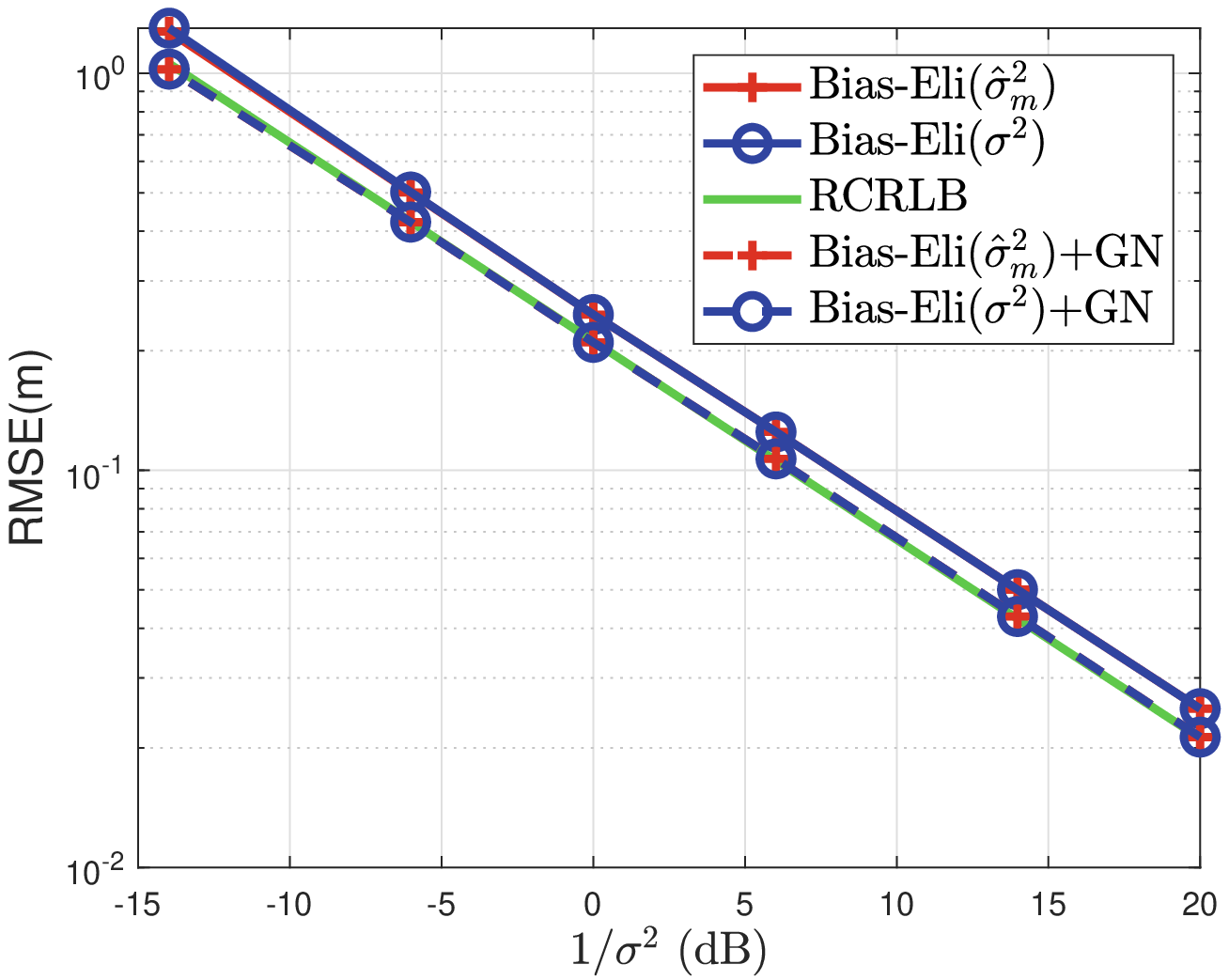}
	\caption{$T=100$}
	\label{NE_vs_BE_varied_sigma}
\end{subfigure}
\caption{Comparison with Bias-Eli estimator using the true noise variance.}
\label{NE_vs_BE}
\end{figure*}

Next, we set ${\bf x}^o=[51,51,51]^\top$ and compare the asymptotic RMSEs under varying noise intensities in the large sample case. We let $T=100$, and $\sigma$ is set as $0.1$, $0.2$, $0.5$, $1$, $2$, and $5$ respectively. For each choice of $\sigma$, $1000$ Monte Carlo tests are executed to evaluate RMSEs. Besides our algorithm, we also perform a one-step GN iteration for the compared algorithms. The result is presented in Fig.~\ref{asymptotic_rmse_vs_noise} where we take the x-axis as $10\log_{10}(1/\sigma^2)$. In the small noise region, all estimators can achieve the RCRLB. This phenomenon arises from the fact that when the intensity of noises is small, the bias of these estimators is negligible, and the RMSE is dominated by their covariance. With the increase in noise intensity, the bias plays a more important role. As a result, the RMSE of the compared algorithms which do not appropriately eliminate the bias deviates from the RCRLB in the large noise region. However, owing to elaborate bias elimination, our estimator achieves the RCRLB consistently irrespective of the noise intensity. 

\begin{figure}[!htbp]
	\centering
	\includegraphics[width=0.48\textwidth]{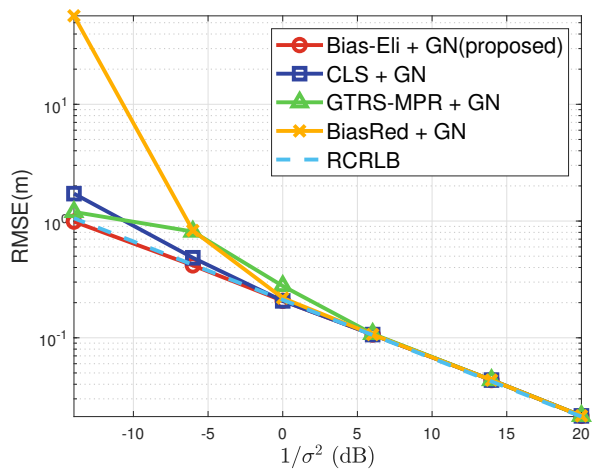}
	\caption{Asymptotic RMSE comparison under varying noise intensities.}
	\label{asymptotic_rmse_vs_noise}
\end{figure}

Till now, we have used the estimated noise variance $\hat \sigma^2_m$ in our proposed Bias-Eli estimator. Here, we compare it (denoted as Bias-Eli($\hat \sigma^2_m$) in Fig.~\ref{NE_vs_BE}) with the Bias-Eli solution utilizing the true noise variance $\sigma^2$ (denoted as Bias-Eli($\sigma^2$) in Fig.~\ref{NE_vs_BE}). Since the estimator Bias-Eli($\sigma^2$) does not need to estimate the noise variance, it could have smaller RMSEs than Bias-Eli($\hat \sigma^2_m$). However, it is not the case in our simulation as shown in Fig.~\ref{NE_vs_BE}, where Fig.~\ref{NE_vs_BE_varied_m} and~\ref{NE_vs_BE_varied_sigma} present the RMSE under varying measurement numbers ($\sigma=2$) and varying noise intensities ($T=100$), respectively. In both cases, the two estimators have negligible differences, which shows the superiority of the proposed noise variance estimation algorithm. 

\begin{table}[htbp]
	\centering
	\caption{CPU Time comparison among different algorithms. Each algorithm is executed $1000$ times to compute the average CPU time. The unit adopted is seconds.} \label{time_comparison}
	\resizebox{0.62\columnwidth}{!}{%
  \begin{tabular}{c c c c c c c}
			\hline\hline
			& $T=1$ & $T=3$ & $T=10$ & $T=30$ & $T=100$ & $T=300$ \\
			\hline
			CLS & 0.2903  &  0.2910  &  0.2924  &  0.2923  &  0.2956  &  0.3056 \\
			\hline
			GTRS-MPR & 0.00047  &  0.00052  &  0.0013  &  0.0065  &  0.1851 &  5.1265 \\
			\hline
			BiasRed & 0.00018  &  {\bf 0.00024}  &  0.0010  &  0.0104  &  0.3858 & 13.0656 \\
			\hline
			{\bf Bias-Eli+GN} & {\bf 0.00015}  &  0.00025  &  {\bf 0.00063}  &  {\bf 0.0016}  &  {\bf 0.0037}  &  {\bf 0.0102} \\
			\hline
		\end{tabular}
	}
\end{table}

Last, we compare the CPU time of different algorithms. All algorithms are executed in Matlab codes, and the CPU type is Intel Core i7-10700. For the compared three estimators, we directly use the open-source codes provided by the authors. Each algorithm is executed $1000$ times to compute the average CPU time, and the result is listed in Table~\ref{time_comparison}. We see that the BiasRed estimator costs the least or second least time when $T$ is small. However, as $T$ increases, its CPU time grows dramatically, and it becomes the most time-consuming one when $T$ exceeds $100$. The CPU time of our proposed algorithm is constantly the least, except for the case of $T=3$ being the second least. In addition, it increases linearly with respect to $T$, coinciding with the theoretical analysis. In the case of $T=300$, i.e., $m=3000$, our algorithm can still achieve a rate of $100$ Hz and is suitable for real-time applications. Similar to the BiasRed method, the CPU time of the GTRS-MPR algorithm also exhibits a nonlinear increasing trend. The CPU time of the CLS algorithm seems to be constant. This is because it involves solving two linear matrix inequalities, which are realized using a CVX toolbox. The calling of the CVX toolbox is heavily time-consuming and dominates the CPU time. In summary, our algorithm has an advantage in time complexity and is especially desirable in the large sample case.

\section{Conclusion}
In this paper, we have proposed a consistent and asymptotically efficient localization estimator based on range-difference measurements. The existence of such an estimator is guaranteed by some readily-checked conditions on measurement noises and sensor deployment. By noting the convergence of the negative log-likelihood function, we first obtained a $\sqrt{m}$-consistent solution and then applied GN iterations to refine it. Specifically, the $\sqrt{m}$-consistent solution is calculated via bias elimination. By solving a $3$-order polynomial, we obtained a consistent estimate of noise variance, which forms the foundation of a consistent bias estimate. Following the preliminary consistent solution, we showed from both theoretical and experimental aspects that a one-step GN iteration would suffice to attain asymptotic efficiency. In addition, the proposed algorithm has $O(m)$ time complexity and costs much less time than the compared algorithms in the large sample case. 

\appendices

\section{Proof of Theorem~\ref{consistent_noise_est}} \label{proof_consistent_noise_est}
Let ${\bf C} (z)={\bf Q}-{\bf S} (z)$. The value $\lambda_{\rm max}({\bf Q}^{-1} {\bf S}(z))$ largely depends on the property of ${\bf C}(z)$, which is summarized in the following lemma.
\begin{lemma} \label{lemma_largest_eig}
	The value $\lambda_{\rm max}({\bf Q}^{-1} {\bf S}(z))$ depends on the eigenvalues of ${\bf C}(z)$. Specifically, 
	\begin{enumerate}
		\item [$(i).$] If ${\bf C}(z)$ is positive definite, then $\lambda_{\rm max}( {\bf Q}^{-1} {\bf S}(z))<1$.
		\item [$(ii).$] If ${\bf C}(z)$ is indefinite, then $\lambda_{\rm max}({\bf Q}^{-1} {\bf S}(z))>1$.
		\item [$(iii).$] If ${\bf C}(z)$ is positive semi-definite, then $\lambda_{\rm max}( {\bf Q}^{-1} {\bf S}(z))=1$.
	\end{enumerate}
\end{lemma}
\begin{proof}
	First, when ${\bf C}(z)$ is positive definite, suppose $\lambda_{\rm max}({\bf Q}^{-1} {\bf S}(z))\geq 1$, i.e., $\lambda_{\rm max}({\bf Q}^{-\frac{1}{2}} {\bf S}(z) {\bf Q}^{-\frac{1}{2}})\geq 1$. Let ${\bf v}={\bf Q}^{\frac{1}{2}} {\bf y}$ be an eigenvector associated with $\lambda_{\rm max}({\bf Q}^{-\frac{1}{2}} {\bf S}(z) {\bf Q}^{-\frac{1}{2}})$. Since ${\bf Q} ={\bf C} (z)+{\bf S} (z)$, we have 
	\begin{align*}
		\|{\bf v}\|^2 & ={\bf v}^\top {\bf Q}^{-\frac{1}{2}} {\bf C}(z)  {\bf Q}^{-\frac{1}{2}} {\bf v} + {\bf v}^\top {\bf Q}^{-\frac{1}{2}} {\bf S}(z) {\bf Q}^{-\frac{1}{2}} {\bf v} \\
		& \geq {\bf y}^\top {\bf C}(z) {\bf y}+ \|{\bf v}\|^2 \\
		& > \|{\bf v}\|^2,
	\end{align*}
	which leads to a contradiction. Thus, $\lambda_{\rm max}( {\bf Q}^{-1}  {\bf S}(z))<1$.
	
	Secondly, when $ {\bf C}(z)$ is indefinite, suppose $\lambda_{\rm max}( {\bf Q}^{-1}  {\bf S}(z)) \leq 1$. There exists a ${\bf y} \neq 0$ such that ${\bf y}^\top  {\bf C}(z) {\bf y}<0$. Let ${\bf v}= {\bf Q}^{\frac{1}{2}} {\bf y}$, as a result, 
	\begin{align*}
		\|{\bf v}\|^2 & ={\bf y}^\top  {\bf C}(z) {\bf y} + {\bf v}^\top  {\bf Q}^{-\frac{1}{2}}  {\bf S}(z)  {\bf Q}^{-\frac{1}{2}} {\bf v} \\
		& < \|{\bf v}\|^2,
	\end{align*}
	which leads to a contradiction. Thus, $\lambda_{\rm max}( {\bf Q}^{-1}  {\bf S}(z))>1$.
	
	Finally, when $ {\bf C}(z)$ is positive semi-definite and singular, by using a similar argument with that of the positive definite case, we obtain  $\lambda_{\rm max}( {\bf Q}^{-1}  {\bf S}(z))\leq 1$. Further, there exists a ${\bf y} \neq 0$ such that $ {\bf C}(z) {\bf y}=0$. Noting that $ {\bf C} (z)= {\bf Q}- {\bf S} (z)$, we have $ {\bf Q}^{-1}  {\bf S}(z) {\bf y}={\bf y}$, which implies $\lambda_{\rm max}( {\bf Q}^{-1}  {\bf S}(z)) = 1$ and completes the proof.       
\end{proof}

We first consider the asymptotic case. Let ${\bf Q}_{\infty}=\lim_{m \rightarrow \infty} {\bf Q}$, ${\bf S}_{\infty}(z)=\lim_{m \rightarrow \infty} {\bf S}(z)$, and ${\bf C}_{\infty}(z)=\lim_{m \rightarrow \infty} {\bf C}(z)$. We have the following lemma. 
\begin{lemma} \label{asymptotic_minimum_solu}
	It holds that  ${\rm min}\{z \mid \lambda_{\rm max}( {\bf Q}_{\infty}^{-1}  {\bf S}_{\infty}(z))=1\}=\sigma^2$.
\end{lemma}
\begin{proof}
	Based on Lemma~\ref{property_of_bounded_variance}, we have
	\begin{equation*}
		\frac{{\tilde {\bf A}}^\top \tilde {\bf A}}{m}= \frac{\tilde {\bf A}^{o \top} \tilde {\bf A}^o}{m} + {\bf S}(\sigma^2) + O_p \left(\frac{1}{\sqrt{m}}\right), 
	\end{equation*}
	where 
	\begin{equation*}
		\tilde {\bf A}^o=\begin{bmatrix}
			-2{\bf a}_1^\top & 1 & -2d_1^o  & -2{\bf a}_1^\top {\bf x}^o+\sigma^2-2d_1^o \|{\bf x}^o\| \\
			\vdots      & \vdots & \vdots & \vdots\\
			-2{\bf a}_m^\top & 1 & -2d_m^o & -2{\bf a}_m^\top {\bf x}^o+\sigma^2-2d_m^o \|{\bf x}^o\|
		\end{bmatrix}
	\end{equation*}
	is the noise-free counterpart of $\tilde {\bf A}$. Since the last column of $\tilde {\bf A}^o$ is a linear combination of the former columns, $\tilde {\bf A}^{o \top} \tilde {\bf A}^o/m$ is singular, so is the asymptotic case, i.e., ${\bf C}_{\infty}(\sigma^2)={\bf Q}_{\infty} -{\bf S}_{\infty}(\sigma^2)=\lim_{m \rightarrow \infty} \tilde {\bf A}^{o \top} \tilde {\bf A}^o/m$ is singular. 
	Hence, $\lambda_{\rm max}({\bf Q}_{\infty}^{-1}{\bf S}_{\infty}(\sigma^2))=1$ based on Lemma~\ref{lemma_largest_eig}. For $z<\sigma^2$, 
	\begin{align*}
		{\bf C}_{\infty}(z) & = {\bf Q}_{\infty} - {\bf S}_{\infty}(z) \\
		& = {\bf C}_{\infty}(\sigma^2)+{\bf S}_{\infty}(\sigma^2)-{\bf S}_{\infty}(z) \\
		& = {\bf C}_{\infty}(\sigma^2) + (\sigma^2-z) \begin{bmatrix}
			{\bf 0}_{(n+1) \times (n+1)} & {\bf 0}_{(n+1) \times 2} \\
			{\bf 0}_{2 \times (n+1)}  & {\bf R}(\sigma^2+z) 
		\end{bmatrix},
	\end{align*}
	in which ${\bf R}(z)=\begin{bmatrix}
	4 & -4 \overbar {d^o} \\
	-4 \overbar {d^o} & 4 (\overbar {{d^o}^2}+\sigma^2) -2z
	\end{bmatrix}$, and 
	\begin{align*}
		{\rm det}({\bf R}(\sigma^2+z)) &  = 16 (\overbar{{d^o}^2}+\sigma^2)-8(\sigma^2+z)-16 {\overbar {d^o}}^2 \\
		& = 16 (\overbar{{d^o}^2}-{\overbar {d^o}}^2)+8(\sigma^2-z) \\
		& >0,
	\end{align*}
	where $\overbar {d^o}=\sum_{i=1}^{m} d^o_i/m$ and $\overbar {{d^o}^2}=\sum_{i=1}^{m} {d^o_i}^2/m$. Therefore, ${\bf R}(\sigma^2+z) \succ 0$ for $z<\sigma^2$, and ${\bf S}_{\infty}(\sigma^2)-{\bf S}_{\infty}(z) \succeq k {\bf S}_{\infty}(\sigma^2)$ for a small $k \in (0,1)$. In addition, 
	\begin{equation*}
		{\bf C}_{\infty}(\sigma^2)+k {\bf S}_{\infty}(\sigma^2) = (1-k) {\bf C}_{\infty}(\sigma^2) +k {\bf Q}_{\infty} \succ 0
	\end{equation*}
	for every $k \in (0,1)$. As a result, 
	\begin{align*}
		{\bf C}_{\infty}(z) & ={\bf C}_{\infty}(\sigma^2)+{\bf S}_{\infty}(\sigma^2)-{\bf S}_{\infty}(z) \\
		& \succeq {\bf C}_{\infty}(\sigma^2)+ k {\bf S}_{\infty}(\sigma^2) \\
		& \succ 0.
	\end{align*}
	That is, for $z<\sigma^2$, $\lambda_{\rm max}({\bf Q}_{\infty}^{-1}{\bf S}_{\infty}(z))<1$ (based on Lemma~\ref{lemma_largest_eig}), which implies $\sigma^2={\rm min}\{z \mid \lambda_{\rm max}( {\bf Q}_{\infty}^{-1}  {\bf S}_{\infty}(z))=1\}$. 
\end{proof}

In the finite-sample case, the two eigenvalues of ${\bf S}(z)$ associated with ${\bf S}_{22}(z)$ are
\begin{align*}
	\lambda_l & = -(z^2-2z-2\overbar{d^2}z) - \sqrt{\Delta}, \\
	\lambda_u & = -(z^2-2z-2\overbar{d^2}z) + \sqrt{\Delta},
\end{align*}
where $\Delta=(z^2-2z-2\overbar{d^2}z)^2+8z^3+16{\bar d}^{~2} z^2-16\overbar{d^2} z^2$.
As $z$ increases, $\lambda_l$ goes to negative infinity, while $\lambda_u$ tends to positive infinity. When $\lambda_u>\lambda_{\rm max}( {\bf Q})$, $ {\bf C} (z)= {\bf Q}- {\bf S} (z)$ is indefinite. Further combining the facts that $ {\bf C} (0)= {\bf Q}$ is positive definite and the eigenvalues of $ {\bf C} (z)$ are continuous with respect to $z$, there exists a $z_m>0$ such that $ {\bf C} (z_m) \succeq 0$. Thus, the set $\{z \mid \lambda_{\rm max}( {\bf Q}^{-1}  {\bf S}(z))=1\}$ is non-empty, and $\hat \sigma^2_m = {\rm min} \{z \mid \lambda_{\rm max}({\bf Q}^{-1} {\bf S}(z))=1 \}$ exists. 
Since $\bf Q$ and ${\bf S}(z)$ converge to ${\bf Q}_{\infty}$ and ${\bf S}_{\infty}(z)$ with the rate $O_p(1/\sqrt{m})$ based on Lemma~\ref{property_of_bounded_variance} and ${\rm min} \{z \mid \lambda_{\rm max}({\bf Q}^{-1} {\bf S}(z))=1 \}$ is a continuous function of $\bf Q$ and ${\bf S}(z)$, the estimate $\hat \sigma^2_m$ shares the same rate of convergence, i.e., $\hat \sigma^2_m-\sigma^2=O_p(1/\sqrt{m})$, which completes the proof.

%
%
%% use section* for acknowledgment
%\section*{Acknowledgment}
%
%
%The authors would like to thank...
%

\ifCLASSOPTIONcaptionsoff
  \newpage
\fi

\small
\bibliographystyle{IEEEtran}
\bibliography{sj_reference}
\end{document}